\documentclass[10pt]{article}
\usepackage[utf8]{inputenc}
\usepackage{macros}
\usepackage{thmtools}
\usepackage{thm-restate}
\usepackage{tablefootnote}

\declaretheorem[name=Primitive, style=remark]{primitive}
\newcommand{\prim}[1]{\hyperref[prim:#1]{Primitive~\ref*{prim:#1}}}
\declaretheorem[name=Assumption, style=remark]{assumption}
\newcommand{\assum}[1]{\hyperref[assum:#1]{Assumption~\ref*{assum:#1}}}

\title{Universal Quantum Speedup for Branch-and-Bound, Branch-and-Cut, and Tree-Search Algorithms}
\author{Shouvanik Chakrabarti }
\author{Pierre Minssen }
\author{Romina Yalovetzky}
\author{Marco Pistoia}
\affil{Global Technology Applied Research, JPMorgan Chase \& Co.}
\date{}

\mathchardef\mhyphen="2D

\newcommand{\bnb}{\textrm{Branch-and-Bound}}
\newcommand{\bnc}{\textrm{Branch-and-Cut}}
\newcommand{\branch}{\mathbf{branch}}
\newcommand{\cost}{\mathbf{cost}}

\newcommand{\qtsize}[1]{\mathbf{QuantumTreeSize}\left({#1}\right)}
\newcommand{\qtsearch}[1]{\mathbf{QuantumTreeSearch}\left({#1}\right)}
\newcommand{\qtminleaf}[1]{\mathbf{QuantumMinimumLeaf}\left(#1\right)}
\newcommand{\qsubtree}[2]{\mathbf{QSubtree}_{#1}\left(#2\right)}
\newcommand{\qsubtreelocal}[1]{\mathbf{QSubtreeLocal}\left(#1\right)}

\newcommand{\iqbb}{\mathbf{Incremental\mhyphen Quantum\mhyphen Branch\mhyphen and\mhyphen Bound}}
\newcommand{\iqbbtext}{\text{Incremental-Quantum-Branch-and-Bound}}
\newcommand{\iqbc}{\mathbf{Incremental\mhyphen Quantum\mhyphen Branch\mhyphen and\mhyphen Cut}}

\newcommand{\iqts}{\mathbf{Incremental\mhyphen Quantum\mhyphen Tree\mhyphen Search}}

\newcommand{\trunc}[1]{\mathbf{TruncateCost}\left(#1\right)}
\newcommand{\ptrunc}[1]{\mathbf{TruncateParentCost}\left(#1\right)}
\newcommand{\twotrunc}[1]{\mathbf{TwoSidedTruncate}\left(#1\right)}
\newcommand{\kthcost}[1]{\mathbf{Kthcost}\left(#1\right)}
\newcommand{\nextcost}[1]{\mathbf{MinimumNextCost}\left(#1\right)}
\newcommand{\curr}{\mathit{current}}
\newcommand{\done}{\mathit{done}}
\newcommand{\heur}{\mathbf{heur}}
\newcommand{\hcost}{\mathbf{hcost}}
\newcommand{\hlocal}{\mathbf{hlocal}}
\newcommand{\hparent}{\mathbf{hparent}}
\newcommand{\incumbent}{\textit{incumbent}}
\newcommand{\bestbound}{\textit{best-bound}}
\newcommand{\cp}{\mathbf{cp}}
\newcommand{\f}{\mathbf{f}}
\newcommand{\degree}{\mathit{deg}}

\SetArgSty{textnormal} %

\begin{document}

\maketitle

\begin{abstract}
    \noindent Mixed Integer Programs (MIPs) model many optimization problems of interest in Computer Science, Operations Research, and Financial Engineering. Solving MIPs is NP-Hard in general, but several solvers have found success in obtaining near-optimal solutions for problems of intermediate size. Branch-and-Cut algorithms, which combine Branch-and-Bound logic with cutting-plane routines, are at the core of modern MIP solvers.
    Montanaro proposed a quantum algorithm with a near-quadratic speedup compared to classical Branch-and-Bound algorithms in the worst case, when every optimal solution is desired. In practice, however, a near-optimal solution is satisfactory, and by leveraging tree-search heuristics to search only a portion of the solution tree, classical algorithms can perform much better than the worst-case guarantee.
    In this paper, we propose a quantum algorithm, Incremental-Quantum-Branch-and-Bound, with \textit{universal} near-quadratic speedup over classical Branch-and-Bound algorithms for \emph{every} input, i.e., if a classical Branch-and-Bound algorithm has complexity $Q$ on an instance that leads to solution depth $d$, Incremental-Quantum-Branch-and-Bound offers the same guarantees with a complexity of $\tilde{O}(\sqrt{Q}d)$. Our results are valid for a wide variety of search heuristics, including depth-based, cost-based, and $A^{\ast}$ heuristics.  Corresponding universal quantum speedups are obtained for Branch-and-Cut as well as general heuristic tree search. Our algorithms are directly comparable to MIP solving routines in commercial solvers, and guarantee near quadratic speedup whenever $Q \gg d$. We use numerical simulation to verify that $Q \gg d$ for typical instances of the Sherrington-Kirkpatrick model, Maximum Independent Set, and Mean-Variance Portfolio Optimization; as well as to extrapolate the dependence of $Q$ on input size parameters. This allows us to project the typical performance of our quantum algorithms for these important problems. 
\end{abstract}

\section{Introduction}
\label{sec:intro}

An \textit{Integer Program} (IP) is a mathematical optimization or feasibility problem in which some or all of the variables are restricted to be integers. A problem where some of the variables are continuous is often referred to as a \textit{Mixed Integer Program} (MIP). IPs and MIPs are ubiquitous tools in many areas of Computer Science, Operations Research, and Financial Engineering. They have been used to model and solve problems such as combinatorial optimization~\cite{papadimitrou1998combinatorial}, Hamiltonian ground-state computation~\cite{baccari2020verifying,Wang2022}, network design~\cite{grotschel1990integer,grotschel1995polyhedral}, resource analysis~\cite{zoltners1980integer}, and scheduling~\cite{sousa1992time,van1999polyhedral}, as well as various problems in finance, including cash-flow management, combinatorial auctions, portfolio optimization with minimum transaction levels or diversification constraints~\cite{cornuejols2006optimization}. Integer Programming is NP-Hard in general, but modern solvers~\cite{gurobi,cplex} are often quite successful at solving practical problems of intermediate sizes, which has led to the development of the aforementioned applications.

The practicality of IP solvers is due to a host of techniques that are employed to bring down the program's execution time, which remains exponential in most cases, but often with significant speedups over naive methods. Such techniques include rounding schemes, cutting-plane methods, search and backtracking schemes, Bender decomposition techniques, and \bnb{} algorithms. There is some overlap between these methods, which are often used in combination for best performance. In recent years, \bnc{} algorithms~\cite[Chapter 11]{cornuejols2006optimization} have emerged as the central technique for solving MIPs. Notably, the core method for the two most common commercially available MIP solvers, Gurobi~\cite{gurobi} and CPLEX~\cite{cplex}, is a \bnc{} procedure. The phrase \bnc{} refers to the combination of two techniques, as \bnc{} is a \bnb{}~\cite{land2022doig} algorithm augmented with cutting-plane methods. \bnc{} algorithms have the desirable feature that while they proceed, they maintain a bound on the \emph{gap} between the quality of the solution found at any point and the optimal solution. This allows for early termination whenever the gap falls to 0 or below a precision parameter input by the user.

Due to the ubiquity and usefulness of \bnc{} algorithms, it is natural to ask whether quantum algorithms can provide a provable advantage over classical algorithms for \bnb{} or, more generally, \bnc{}. In recent years, quantum algorithms have been shown to be provably advantageous for many continuous optimization problems, including Linear Systems~\cite{harrow2009linear}, Zero-Sum Games~\cite{li2019sublinear}, Generalized Matrix Games~\cite{li2021sublinear}, Semi-definite Programs~\cite{brandao2017quantum}, and General Zeroth-Order Convex Optimization~\cite{van2020convex,chakrabarti2020quantum}, as well as some discrete optimization problems, such as Dynamic Programming~\cite{ambainis2019dynamic} and Backtracking~\cite{montanaro2018backtracking}. There are also conditional provable speedups for some interior point methods, for example for Linear Programs (LPs), Semi-Definite Programs (SDPs)~\cite{kerenidis2018sdp}, and Second-Order Cone Programs (SOCPs)~\cite{kerenidis2018sdp}.

Montanaro~\cite{montanaro2020branch} proposed an algorithm with a provable near-quadratic quantum speedup over the worst-case classical \bnb{} algorithms where the task is to ensure that \emph{all} optimal solutions are found. However, this quantum algorithm does not accommodate early stopping based on the gap, thereby not accounting for the fact that, in most practical situations, one is satisfied with any solution whose quality is within some precision parameter of the optimal. Therefore, the true performance of classical \bnb{} on a problem can be very different from the worst case considered by Montanaro. Furthermore, simply \emph{estimating} the worst-case classical performance can be much more computationally intensive than simply finding an approximate solution to an IP. This makes it difficult to know whether the quantum advantage holds for any practical problem, as well as to estimate what the true expected performance of the quantum algorithm would be.

Therefore, we look for a stronger notion of quantum speedup, which we term ``universal speedup''. We say that a quantum algorithm for a task has a \textit{universal speedup} over a classical algorithm if the speedup is over the \emph{actual performance} of the classical algorithm in \emph{every case}. Informally, this means that if a classical algorithm gets ``lucky'' on a particular family of inputs, so does the associated quantum algorithm. Universal speedup also allows the true performance of the quantum algorithm to be empirically inferred from that of the classical variant, as every classical execution has a provably more efficient quantum variant.

Ambainis and Kokainis~\cite{ambainis2017tree} gave a universal speedup for the backtracking problem. The question of whether such a speedup is possible for $\bnb{}$ has been left open in Montanaro's work \cite{montanaro2020branch}. The main goal of this paper is to shed a light on this question. 

\subsection{Classical \bnc{}}
\label{sec:classical-bnc}
\subsubsection{\bnb{}} 
While there are many possible \bnb{} algorithms (see~\cite{clausen1999branch,morrison2016branch} for a survey of techniques), they all follow a formally similar approach, where an input MIP is solved using a tree of \emph{relaxed} optimization problems that can be solved efficiently: each of them yields a solution that may not be feasible. If the solution is not feasible, it is then used to produce new relaxed sub-problems, each of which  produces a solution of worse quality than its parent, but is chosen so that the solution obtained is more likely to be feasible for the original MIP. Eventually, this process leads to feasible solutions, which form the leaves of the explored tree. Classical \bnb{} methods search for the leaf with the best solution quality.

Treating the relaxed problems as the nodes of a tree allows a \bnb{} algorithm to be described as exploring a \bnb{} tree, specified by the following parameters:
\begin{enumerate}
    \item The root $N_0$ of the tree
    \item A branching function $\branch$ that for any input node $N$ returns the set of the children of $N$
    \item A cost function $\cost$ mapping any node to a real number
\end{enumerate}
For a tree $\CT$ to be a valid \bnb{} tree, it must be finite, and satisfy the so called \emph{\bnb{} condition}, which states that for any node $N \in \CT$, we have
$\mathbf{cost}(N) \le \min_{N' \in \branch(N)}\{\cost(N')\}$,
i.e., any node has lower cost than its children. Note that these conditions simply formalize the description above: we choose $\mathbf{cost}$ to capture the quality of the solution at a node (where lower is better). In this correspondence, a node represents a more restricted subproblem than its parent and must therefore have a higher cost.

Given this definition, we have the following abstract formulation of the \bnb{} problem:
\begin{definition}[Abstract \bnb{} Problem]
\label{defn:abstract-bnb}
Given a tree $\mathcal{T}$ with root $N_0$, error parameter $\epsilon$, and oracles $\branch,\cost$ satisfying the \bnb{} condition, return a node $N$ so that $N$ is a leaf, and for any leaf $L$ of $\CT$, $\cost(N) \le \cost(L) + \epsilon$.
\end{definition}

An example of a \bnb{} specification is in Mixed Integer Programming with a linear, convex quadratic, or general convex function. This situation is common in financial engineering, e.g., for portfolio optimization. Consider an MIP $\mathcal{I}$ whose objective is to minimize a function $\mathbf{f} \colon \mathcal{F} \to \R$, where $\mathcal{F} \subseteq \R^{D_1} \times \Z^{D_2}$ is specified by a set of $m$ constraints. A relaxed problem can be obtained by expanding the domain of optimization, for example by removing the integrality constraints on the last $D_2$ variables. The resulting relaxed problem can be solved efficiently, i.e., in time $\poly(m,D)$, where $D = D_1 + D_2$. Each node of the \bnb{} tree contains such a relaxed problem. Specifically, for each node $N$, $\cost(N)$ is defined as the value of the objective at the solution of the relaxed problem at $N$. We generate sub-problems as follows. Let the optimal solution $\vec{x}$ at $N$ have a coordinate $x_j : j > D_1$ with value $c \notin \Z$. We obtain new problems by adding either $x_j \le \lfloor c \rfloor$ or $x_j \ge \lceil c \rceil$ as constraints to the relaxed problem at $N$.

A \bnb{} algorithm proceeds by searching the \bnb{} tree using a search heuristic $\CH$ to find the leaf with minimum cost. We use the following terminology in the rest of the paper. The nodes of the tree are grouped into three categories: ``undiscovered'', ``active'', and ``explored'':
\begin{enumerate}
    \item \emph{Undiscovered} nodes are those that the algorithm has not seen yet, i.e., they have not been returned as the output of any $\branch$ oracle call.
    \item An \emph{active} node is one that has been discovered by the algorithm, but has not had the $\branch$ oracle called on it, i.e., the node has been discovered, but its children have not.
    \item An \emph{explored} node is one that has been discovered and whose children have also been discovered by calling the $\branch$ oracle on the node itself.
\end{enumerate}
With this setup, we can describe the algorithm as follows for a tree $\CT$: 

\begin{enumerate}
    \item The list of active nodes $\verb|active|$ is initialized to contain only the root node, $N_0$.
    \item The next candidate for exploring is returned by applying the search heuristic, $\CH$, to the list $L_{\mathrm{active}}$. Once a node is explored, it is removed from $L_{\mathrm{active}}$ and its children are added.
    \item During exploration, two quantities are maintained:
    \begin{enumerate}[nosep]
        \item The \emph{incumbent}: the minimum cost of any discovered leaf. This is an upper bound on the possible value of the optimal solution.
        \item The \emph{best-bound}: the minimum cost of any active node. Any newly discovered node will have a greater cost than this so this is a lower bound on the optimal solution.
    \end{enumerate}
    \item The gap between the incumbent value and the best-bound is an upper bound on the sub-optimality of the incumbent solution. Thus, when the gap falls below some $\epsilon$, the incumbent can be returned as an $\epsilon$-approximate solution.
\end{enumerate}

\subsubsection{Cutting Planes}
Consider a MIP with feasible set $\CF$ and let $\CF'$ be the feasible set with its integrality constraints relaxed. A \textit{cutting plane} is a hyperplane constraint that can be added, such that the true feasible set $\CF$ is unchanged, but  some infeasible solutions are removed from the relaxed feasible set $\CF'$. A cutting-plane algorithm repeatedly adds cutting planes to the relaxed MIP until the solution to the relaxed problem is feasible for the original (unrelaxed) MIP. Since no such solutions are removed due the cutting planes, the solution thus obtained is optimal. \bnc{} combines cutting planes with \bnb{}. As the \bnb{} tree $\CT$ is searched, cutting-plane methods are used at various nodes to obtain new cutting planes. A cutting plane at a node $N$ is \emph{local} if it is only a valid cutting plane for the sub-tree of $\CT$ rooted at $N$, and \emph{global} if it applies to the whole tree. A classical \bnc{} problem simply searches the \bnb{} tree as discussed previously and, when a cutting plane is found, adds it to the constraints.

\subsubsection{Search Heuristics}
A \bnb{} algorithm requires a search heuristic $\CH$ to specify which active node is to be explored next. For a heuristic $\CH$ to be practically usable, it must have an implicit description that does not require explicitly writing down a full permutation over $T$ elements (where $T$ is often exponential in problem-size parameters). 

In the rest of the paper, we restrict ourselves to the class of search heuristics consisting of those that rank nodes based on a combination of \textit{local information} (that can be computed by calling $\branch,\cost$ at the node itself) and information that can be obtained from the output of $\branch,\cost$ at the parents of the node. We define the corresponding family of heuristics as follows:
\begin{definition}[Branch-Local Heuristics]
\label{defn:branch-local}
 A \emph{Branch-Local Heuristic} for \bnb{} tree, is any heuristic that ranks nodes in ascending order of the value of some function $\heur$, where for any node $N$ at depth $d(N)$ in the tree, with path $r \to n_1 \to \dots \to n_{d(N)-1} \to N$ from the root, 
 \begin{equation}
    \heur(N) = f(\hlocal(N),\hparent(r),\hparent(n_1)\dots,\hparent(n_{d(n)-1}),d(N))
\end{equation}
where each of $\hlocal,\hparent$ makes a constant number of queries to $\branch,\cost$, and $\f$ is a function with no dependence on $N$.
\end{definition}
The above definition captures all commonly used search heuristics that we are aware of, including the three following main categories:
\begin{enumerate}
    \item \textbf{Depth-first heuristics} The heuristic $h$ is specified by a sub-function $h' \colon \CN \mapsto  \CS(2)$  that orders the children of any node. The tree is explored via a depth-first search where among two children of the same node, the child ordered first by $h'$ is explored first.
    \item \textbf{Cost-based heuristics} The heuristic ranks nodes simply by the value of the $\cost$ function. Despite its simplicity, this is the most commonly used search heuristic for $\bnb{}$ and is used by both Gurobi and CPLEX~\cite{gurobi,cplex}. A simple generalization is to consider arbitrary functions of the $\cost$.
    \item \textbf{A$^{\ast}$ heuristics} These heuristics rank nodes by the value of $\heur(N) =\hcost(N) + d(N)$, where $\hcost(n)$ makes a constant number of queries to $\branch$ and $\cost$, and $d(N)$ is the depth of $N$.
\end{enumerate}
The family of heuristics in \defn{branch-local} is probably broader than needed to include all search methods used in practice. In particular, it includes all the search strategies for \bnb{} discussed in \cite[Section 3]{morrison2016branch}, and for parallel \bnb{} in \cite{clausen1999best}. Nevertheless, it is possible that some \bnb{} algorithm uses a heuristic that does not fall under this definition and such algorithms will not be covered by our results. Finally, we note that \defn{branch-local} does not require the explored tree to be a \bnb{} tree and can be easily extended to any tree, even one where there is no $\cost$ oracle, in which case we simply take $\cost$ to be a constant. In the presentation, we will therefore often omit explicitly passing $\cost$ to functions unless required.

\subsection{Discussion of Existing Results}
\label{sec:existing}

The primary existing result on \bnb{} algorithms is the following due to Montanaro~\cite{montanaro2020branch}.
\begin{theorem}[\!{\cite[Theorem 1]{montanaro2020branch}}]
\label{thm:montanaro-bnb-existing}
Consider a \bnb{} tree $\mathcal{T}$ specified by oracles $\branch,\cost$. Let $d$ be the depth of $\mathcal{T}$, and the cost of the solution node be $c_{\min}$. Let $c_{\max}$ be a given upper bound on the cost of any explored node. Assume further that $\mathcal{T}$ has $T_{\min}$ nodes with cost $\le c_{\min}$. Then there exists a quantum algorithm solving the \bnb{} problem on $\mathcal{T}$ with failure probability at most $\delta$ using $\tilde{O}\left(\sqrt{T_{\min}} d^{3/2} \log(c_{\max})\left(\frac{1}{\delta^2}\right)\right)$ queries to $\branch,\cost$.
\end{theorem}
We note that, due to the subsequent refinement of a tree search algorithm by Apers, et al.~\cite{apers2019walk}, the query complexity in \thm{montanaro-bnb-existing} is trivially improved to $\tilde{O}\left(\sqrt{T_{\min}} d \log(c_{\max})\left(\frac{1}{\delta^2}\right)\right)$. A classical \bnb{} algorithm that is required to output \emph{all} optimal solutions (a stronger requirement than in \defn{abstract-bnb} must search at least $T_{\min}$ nodes of the \bnb{} tree, thereby making at least $T_{\min}$ queries to $\branch,\cost$. \thm{montanaro-bnb-existing} thus provides a nearly quadratic speedup for this scenario, whenever $d \ll \sqrt{T_{\min}}$. In most practical problems $T_{\min}$ is exponentially larger than $d$ so this condition is satisfied.

The problem of returning a single approximate solution (as in \defn{abstract-bnb}) may not require $O(T_{\min})$ queries to $\branch,\cost$, and it is not clear whether the speedup from \thm{montanaro-bnb-existing} does not automatically apply. In fact, even if the error parameter is 0, the early termination of \bnb{} may be able to find an exact solution without exploring $O(\sqrt{T_{\min}})$ queries, as whenever the gap falls to $0$ the existing incumbent solution has been proved optimal. True worst-case bounds on classical algorithms for \defn{abstract-bnb} are hard to characterize analytically, and infeasible to infer from empirical data. Furthermore, the true performance of the algorithm on practical families of instances can be much (even exponentially) better than the worst case. Thus even if \thm{montanaro-bnb-existing} were to provide a worst case speedup, this may not translate to a speedup for problems of interest. Finally, the algorithm in \thm{montanaro-bnb-existing} is easily augmented with local cutting planes by modifying the branch oracle, but cannot be used with global cutting planes as the whole algorithm has only one stage and cutting planes discovered at a node are therefore unavailable to any nodes outside the subtree containing the descendants of that node.

A closely associated problem to \bnb{} is that of tree search. Backtracking can be abstractly formulated as the problem of finding a marked node in a tree. Ambainis and Kokainis demonstrated a universal speedup for backtracking when the tree is explored using an \emph{depth-first} heuristic, as is common in backtracking algorithms. We state a version of their result (equivalent to \cite[Theorem 7]{ambainis2017tree})
\begin{theorem}
\label{thm:universal-backtracking-bound}
Consider a tree $\CT$ with depth $d$ specified by a $\branch$ oracle, and a marking function $\f$ that maps nodes of the tree to 0 or 1, where a node $n$ is marked if $\f(n) =1$. Assume that at least one node is marked. Then let $\CA$ be a classical algorithm that explores $Q$ nodes ($O(Q)$ queries to $\branch$) of $\CT$ using some depth-first heuristic and outputs a marked node in $\CT$, and otherwise returns $0$. For $0 \le \delta \le 1$, there exists a quantum algorithm that makes $\tilde{O}\left(\sqrt{Q}d^{3/2}\log^2\left(\frac{n \log(Q)}{\delta}\right)\right)$ queries to $\branch$ and with probability at least $1 - \delta$, returns the same output as $\CA$.
\end{theorem}
In contrast, the best known worst-case quantum algorithm~\cite{apers2019walk} for tree search uses $\tilde{O}(\sqrt{T}d \log(1/\delta))$ queries to $\branch$. We note that \thm{universal-backtracking-bound} has a worse dependence on the depth $d$ (by a factor of $\sqrt{d}$). The better $d$ dependence cannot be obtained by a trivial modification of \thm{universal-backtracking-bound} and the question of whether such a dependence is possible is left open in~\cite{apers2019walk}. The universal speedup is also valid only for depth-first heuristics.

\subsection{Contributions of This Work}
\label{sec:results}
\paragraph{Branch-and-Bound.} Our main result is a universal quantum speedup over all classical \bnb{} algorithms that explore \bnb{} trees using Branch-Local Heuristics~(\defn{branch-local}). The quantum algorithm must depend on the heuristic used classically. We therefore introduce a meta-algorithm, $\iqbb$, that transforms a classical algorithm $\mathcal{A}$ into a quantum algorithm $\iqbb_{\mathcal{A}}$, with the following guarantee:

\begin{theorem}[Universal Speedup for Branch-and-Bound]
\label{thm:iqbb-main}
Suppose a classical algorithm $\mathcal{A}$ uses a Branch-Local heuristic $h$ to explore a \bnb{} tree $\CT$ rooted at $r$ and specified by oracles $\branch,\cost$ and returns a solution with $\cost$ at most $\epsilon$ greater than the optimum, using $Q(\epsilon)$ queries to $\branch$ and $\cost$. Then there exists a quantum algorithm $\iqbb_{\CA}$ that solves the problem to the same precision with probability at least $1- \delta$, using $\tilde{O}\left(\sqrt{Q(\epsilon)} d \log(c_{\max}h_{\max})\log^2\left(\frac{\log(T)}{\delta}\right)\right)$ queries to $\branch$ and $\cost$, where $d$ is the depth of $\CT$, $T$ is an upper bound on the size of $\CT$, $c_{\max}$ is an upper bound on the $\cost$ of any node, and $h_{\max}$ is an upper bound on the $\heur$ function associated with $h$ (see \defn{branch-local}).
\end{theorem}

\paragraph{Branch-and-Cut.}
We also show a construction of a universal speedup for \bnc{} by showing how cutting planes can be incorporated into $\iqbb{}{}$. We note that \emph{local} cutting planes can be easily incorporated into both Montanaro's algorithm~\cite{montanaro2020branch} as well as ours by a simple modification of the $\branch$ oracle. Our algorithm additionally admits the incorporation of \emph{global} cutting planes, i.e., cutting planes that are found at an interior node of the \bnb{} tree but are applicable also to nodes in other branches. Note that the addition of global cutting planes lead to the $\branch$ and $\cost$ oracles evolving through the execution. We denote the oracles after $j$ cutting planes have been found as $\branch^{(j)},\cost^{(j)}$. In reasonable problems the cost of all these oracles will be nearly the same so as to maintain the efficiency of solving the relaxed subproblems.

\begin{theorem}[Universal Speedup for \bnc{}]
\label{thm:iqbc-main}
Suppose a classical \bnc{} algorithm $\mathcal{A}$ uses a Branch-Local heuristic $h$, and up to $p$-global cutting planes, to explore a \bnb{} tree $\CT$ rooted at $r$ that is specified by oracles $\branch,\cost$, and returns a solution with $\cost$ at most $\epsilon$ greater than the optimum, using $Q(\epsilon)$ total queries to $\branch^{(j)},\cost^{(j)}$ for $1 \le j \le p$. Then there exists a quantum algorithm $\iqbc_{\CA}$ that solves the problem to the same precision with probability at least $1- \delta$, using $\tilde{O}\left(\sqrt{Q(\epsilon)} d p\log(c_{\max}h_{\max})\log^2\left(\frac{\log(T)}{\delta}\right)\right)$ queries to $\branch^{(j)},\cost^{(j)}$ for $1 \le j \le p$, and $\mathbf{cp}$, where $d$ is the depth of $\CT$, $T$ is an upper bound on the size of $\CT$, $c_{\max}$ is an upper bound on the $\cost$ of any node, $h_{\max}$ is an upper bound on the $\heur$ function associated with $h$ (see \defn{branch-local}), and $\cp$ is a function such that $\mathbf{cp}(n)=1$ if and only if a cutting plane if found at a given node.
\end{theorem}
Observe that \thm{iqbc-main} introduces an additional dependence on $p$ and so the algorithm may not have an advantage when more than an exponential number of global cutting planes are used. Note however, that adding constraints to the relaxed problem makes the runtime at each node longer. The \bnc{} heuristic relies on the solution of each node being efficient, therefore in practice only $\poly(n)$ cutting planes are used. 

Most commercial MIP solvers use routines based on \bnb{} and \bnc{} for which our techniques obtain universal speedup. Thus, our algorithms obtain quantum speedups for any problem where the classical state of the art relies on these routines. In \sec{experiments}, we investigate numerically the degree of speedup we can expect to obtain for three such problems: the Sherrington-Kirkpatrick model, Maximum Independent Set, and Portfolio Optimization with diversification constraints. Another significant example is the problem of finding Low-Autocorrelation Binary Sequences~\cite{Packebusch2016}, where the classical state of the art uses \bnb{} with a cost-based heuristic, and obtains performance that is empirically observed to scale as $O(1.73^n)$, where $n$ is the problem size. Classically the problem has been solved up to $n = 66$. Therefore, $\iqbb$, which obtains a universal near-quadratic advantage over classical \bnb{}, could eventually enable solving problems of sizes  $\sim 100$ or beyond.

\paragraph{Heuristic Tree Search} Finally, our results also yield new algorithms for quantum tree search using Branch-Local search heuristics. A tree search problem is specified by a $\branch$ oracle and a marking function $\f$.
\begin{theorem}[Universal Speedup for Heuristic Tree Search]
\label{thm:iqts-main}
Consider a tree $\CT$ rooted at $r$ that is specified by a $\branch$ oracle, and a marking function $\f$ that marks at least one node.
Suppose a classical algorithm $\mathcal{A}$ uses a Branch-Local heuristic $h$ and returns a marked node $n$ such that $\f(n) =1$, using $Q(\epsilon)$ queries to $\branch$ and $\f$. Then there exists a quantum algorithm $\iqts_{\CA}$ that offers the same guarantee with probability at least $1 -\delta$, using $\tilde{O}\left(\sqrt{Q(\epsilon)} d \log(h_{\max})\log^2\left(\frac{\log(T)}{\delta}\right)\right)$ queries to $\branch$ and $\f$, where $d$ is the depth of $\CT$, $T$ is an upper bound on the size of $\CT$, and $h_{\max}$ is an upper bound on the $\heur$ function associated with $h$ (see \defn{branch-local}).
\end{theorem}
\thm{iqts-main} extends the universal speedup for backtracking~\cite{ambainis2017tree} to many other search heuristic. For the particular case of backtracking, \thm{iqts-main} yields an algorithm for backtracking with better depth dependence than~\thm{universal-backtracking-bound}. For depth-first heuristics, it suffices to have $\heur = O(Q(\epsilon))$. Thus we have an asymptotically faster backtracking algorithm whenever $d = \omega(\log(Q))$.

\section{Proofs of Results}
\label{sec:technical-proofs}
In this section, we prove our primary theoretical results. We structure the section as follows: In \sec{quantum-prelim}, we describe quantum subroutines from existing literature that we use to construct our algorithms. \sec{incremental-quantum} then introduces the framework within which we prove all our main results, and reduces the proof to a subroutine that efficiently estimates subtrees corresponding to the first nodes explored by classical algorithms. Assuming a guarantee~(\thm{subtree-main}) on the correctness and efficiency of this procedure, we prove \thm{iqbb-main},\thm{iqbc-main}, and \thm{iqts-main}, in \sec{iqbb}, \sec{iqbc}, and \sec{iqts}, respectively. Finally in \sec{partial-subtree}, we prove our main technical ingredient in~\thm{subtree-main}, thereby completing the proofs.

\subsection{Quantum Preliminaries}
\label{sec:quantum-prelim}

Our results are based on quantum random walks on tree graphs. However, we shall not directly analyze  quantum walks but will instead present algorithms, algorithm using existing algorithmic primitives as subroutines. We present the primitives below with their associated claims and guarantees.  This will allow us to avoid a direct discussion of the implementations of these subroutines and preliminaries on quantum algorithms.

   \begin{primitive}
   \label{prim:qtsearch}[Quantum Tree Search]
   Suppose we are given a tree $\mathcal{T}$ with $T$ nodes, depth $d$, specified by a branching oracle $\branch{}$, and a marking function $\f$ on the nodes. Let $M = \lbrace x | x \in \CT, \f(x) = 1 \rbrace$ be the set of marked nodes and $m = |M|$. A quantum tree search algorithm $\qtsearch{\cdot}$ is such that $\qtsearch{\branch,d,T,\f,\delta}$  and with failure probability at most $\delta$ returns a member of $M$ if it is not empty, and otherwise returns ``not found''. There are two versions of $\qtsearch{\cdot}$ that both offer the above guarantee.
   \begin{enumerate}
       \item The first, due to Apers, et al.~\cite{apers2019walk} makes $O(\sqrt{T}d\log(d)\log\left(1/\delta)\right)$  queries to $\branch$ and $\f$. This is an improvement of an earlier algorithm due to Montanaro~\cite{montanaro2018backtracking} makes $O(\sqrt{T}d^{3/2}\log(d)\log\left(1/\delta)\right)$  queries to $\branch$ and $\f$.
       \item The second, due to  Jarret and Wan~\cite{jarret2018improved} makes $O\left(\sqrt{Td}\log^4(md)\log(m/\delta)\right)$ queries to $\branch$ and $\f$. This version of the algorithm is more efficient when the number of nodes $m$ is a polynomial function of $d$.
   \end{enumerate}
   \end{primitive}

\begin{primitive}[Quantum Tree Size Estimation]
\label{prim:qtsize}
Suppose we are given a tree $\CT$ specified by an oracle $\branch{}$, which has depth at most $d$, $T$ nodes and maximum degree $\degree = O(1)$, and parameters $T_0$, relative error $\epsilon < 1$, and failure probability $\delta$. There exists a quantum algorithm  $\qtsize{\branch,d,T_0,\delta,\epsilon}$ due to Ambainis and Kokainis~\cite{ambainis2017tree} that makes $O\left(\frac{\sqrt{T_0 d}}{\epsilon^{3/2}}\log^2(1/\delta)\right)$  queries to $\branch$ and
   \begin{enumerate}
       \item If $T \le T_0/(1 + \epsilon)$, returns $\tilde{T} \le T_0$ such that $|\tilde{T} - T| \le \epsilon T$.
       \item If $T \ge T_0(1 + \epsilon)$, outputs ``more than $T_0$ nodes''.
   \end{enumerate}
   We shall use the algorithm instead with one sided error. The error in the output promise can be made one-sided querying the algorithm with $T_0' = T_0(1 + \epsilon)$. We also make the relative error (in case an estimate is returned) one  sided as follows:  the original algorithm returns an estimate $\tilde{T}$ such that $(1 - \epsilon) T \le T/(1 + \epsilon) \le \tilde{T} \le (1 + \epsilon) T$. The modified algorithm instead returns $\tilde{T}(1 + \epsilon)$.
   We therefore have a modified version such that $\qtsize{\branch,d,T_0,\delta,\epsilon}$, uses the same number of queries as before and
   \begin{enumerate}
       \item If $T \le T_0$, returns $\tilde{T} \le T_0(1 + \epsilon)$ such that $T \le \tilde{T} \le T(1 + \epsilon)^2$.
       \item If $T \ge T_0(1 + \epsilon)^{2}$, output ``more than $T_0$ nodes''.
   \end{enumerate}
\end{primitive}

We shall also use a quantum routine to find the minimum cost leaf in a tree, given access to a $\branch{}$ oracle for the tree and a $\cost{}$ oracle for the nodes. We observe here that the \bnb{} algorithm effectively finds the minimum leaf in a tree. We can therefore use a minor variant of Montanaro's algorithm for \bnb{} for this routine (Algorithm 2 from \cite{montanaro2020branch}) using the version of $\qtsearch{\cdot}$ in \prim{qtsearch} and the two-sided version of $\qtsize{\cdot}$ in \prim{qtsize}. The guarantee follows from~\cite[Theorem 1]{montanaro2020branch}.

\begin{primitive}[Quantum Minimum Leaf]
\label{prim:qtminleaf}
Given a \bnb{} tree $\CT$ specified by a $\branch$ oracle and a $\cost$ oracle on the nodes, such that the tree has $T$ nodes, depth $d$, and that the maximum cost of any node in the tree is $c_{\max}$, there exists an algorithm $\qtminleaf{\branch,\cost,d,c_{\max},T,\delta}$ that makes $\tilde{O}\left(\sqrt{T}d\log(c_{\max})\log^{2}\left(\frac{1}{\delta}\right) \right)$ oracle calls to $\branch$ and $\cost$ and with failure probability at most $\delta$ returns the minimum leaf in the tree.
\end{primitive}

\subsection{Incremental Tree Search Framework}
\label{sec:incremental-quantum}
Our main results (see \sec{results}) are all established within the same framework, inspired by the method used for incremental tree search by Ambainis and Kokainis~\cite{ambainis2017tree}. Consider any algorithm that incrementally explores a tree using a heuristic. At any point, the nodes that have been explored form a sub-tree of the original. We define the following problem which we call Quantum Partial Subtree Estimation
\begin{definition}[Quantum Partial Subtree Estimation]
\label{defn:subtree-estimation}
Given a tree $\CT$ rooted at $r$ and specified by a $\branch$ oracle, a classical algorithm $\CA$ that explores it using heuristic $h$ such that $h_{\max}$ is the maximum value for any node, of $\heur$ used by $h$ for ranking (\defn{branch-local}), and a positive integer $m$, a \emph{Quantum Partial Subtree Estimation algorithm} is any algorithm $\qsubtree{\CA}{r,\branch,d,h_{\max}, 2^m,\delta}$ that returns with probability at least $1 - \delta$, a value $x$ that can be used to define a new oracle $\branch'$ satisfying the following guarantees:
\begin{enumerate}
    \item $\branch'$ is a function of $x$ and $\branch$.
    \item $\branch'$ makes a constant number of queries to $\branch$.
    \item The sub-tree corresponding to $\branch$ has $k2^{m}$ nodes.
    \item The sub-tree corresponding to $\branch$ contains the first $2^{m}$ nodes visited by $\mathcal{A}$.
\end{enumerate}
\end{definition}
Ambainis and Kokainis~\cite{ambainis2017tree} developed a Quantum Partial Subtree Estimation for depth first heuristics that uses $\tilde{O}(\sqrt{2^m} d^{3/2})$ queries to $\branch$ and use it to obtain a universal speedup for backtracking. Our main technical ingredient will be to develop a Quantum Partial Subtree Estimation for any Branch-Local Heuristic. We will show the following theorem
\begin{theorem}[General Quantum Subtree Estimation]
\label{thm:subtree-main}
Consider a tree $\CT$ of depth $d$ rooted at $r$ and specified by a $\branch$ oracle. Let $\CA$ be a classical algorithm that searches $\CT$ using Branch-Local search heuristic $h$. Let $h_{\max}$ be an upper bound on the $\heur$ function used by $h$ for ranking, and $m$ a positive integer. There exists an algorithm $\qsubtree{\CA}{r,\branch,d,h_{\max},2^{m},\delta}$ that makes $\tilde{O}\left(d\sqrt{2^m}\log(h_{\max})\log^2\left(\frac{1}{\delta}\right)\right)$ queries and with probability at least $1 - \delta$. satisfies the conditions of \defn{subtree-estimation} with $k=4$. 
\end{theorem}
We defer the proof of \thm{subtree-main} to \sec{partial-subtree}. First, we show how our main results may be obtained as a consequence. We will start with heuristic tree search as it is the simplest demonstration of the incremental search framework.

\subsection{Universal Speedup for Heuristic Tree Search}
\label{sec:iqts}
\begin{algorithm}[!htbp]
\DontPrintSemicolon
\KwIn{$\branch$ oracle of tree $\CT$ of depth $d$ rooted at $r$, marking function $\f$, classical algorithm $\CA$ that explores the tree using heuristic $h$, upper bound $h_{\max}$ on heuristic, upper bound $T$ on size of tree, failure probability $\delta$}
\KwResult{$\iqts_{\mathcal{A}}(r,\branch,d,\f,h_{\max},T,\delta)$ returns node $n \in \CT$ such that $\f(n)=1$.}
Set $m = 0$.\;
\While{$\top$}
{
    Use $\qsubtree{\CA}{r,\branch,d,h_{\max},2^m,\delta/4\log{T}}$ to obtain oracle $\branch_1$.\;
    \If{$\qtsearch{\branch_1,d,4\cdot2^{m},\f,\delta/4\log{T}}$ returns node n}
    {
        \KwRet $n$. \;
    }
        $m \gets m+1$. \;
}
\caption{Incremental Quantum Tree Search}
\label{algo:iqts}
\end{algorithm}

\begin{proof}[\textbf{Proof of \thm{iqts-main}.} ] We show that $\iqts_{\CA}$ defined in \algo{iqts} satisfies the conditions of \thm{iqts-main}. 
First, we ignore failure probabilities. Suppose $\CA$ finds a marked node after $Q$ queries. Consider the loop iteration where $m_Q = \lceil \log(Q) \rceil$. Since the sub-tree corresponding to $\branch',\cost'$ contains at least the first $Q$ nodes explored by $\CA$ and by assumption $\CA$ finds a marked node after exploring at most $Q$ nodes, $\qtsearch{\cdot}$ in this iteration returns a marked node. The algorithm thus terminates after at most $m_Q = \lceil \log(Q) \rceil$ iterations.

At each $m$ we make $\tilde{O}\left(d\sqrt{2^m}\log(h_{\max})\log^2\left(\frac{\log{T}}{\delta}\right)\right)$ queries for $\qsubtree{\CA}{\cdot}$ and $\tilde{O}\left(d\sqrt{4 \cdot 2^m}\log^2\left(\frac{\log{T}}{\delta}\right)\right)$ queries for $\qtsearch{\cdot}$; leading to a total of $\tilde{O}\left(d\sqrt{2^m}\log(h_{\max})\log^2\left(\frac{\log{T}}{\delta}\right)\right)$. Across iteration the number of queries made is therefore,
\begin{equation}
    \tilde{O}\left(\left(d\log(h_{\max})\log^2\left(\frac{\log{T}}{\delta}\right)\right)\sum_{m=0}^{m_Q} \sqrt{2^m}\right) = \tilde{O}\left(d\sqrt{2^{m_Q}}\log(h_{\max})\log^2\left(\frac{\log{T}}{\delta}\right)\right)
\end{equation}
where the equality follows from the formula for the sum of geometric series. Finally, note that there are at most $2m_Q$ possibilities for failure, each with probability $< \delta/4\log(T)$. Thus the total failure probability is at most $\delta$, and the proof is complete.
\end{proof}

\subsection{Universal Speedup for \bnb{}}
\label{sec:iqbb}
\begin{algorithm}[!htbp]
\DontPrintSemicolon
\KwIn{$\branch$ and $\cost$ oracles of $\bnb{}$ tree $\CT$ of depth $d$ rooted at $r$, classical algorithm $\CA$ that explores the tree using heuristic $h$, upper bound $h_{\max}$ on heuristic, upper bound $T$ on size of tree, failure probability $\delta$, precision parameter $\epsilon$}
\KwResult{$\iqbb_{\mathcal{A}}(r,\branch,\cost,d,c_{\max},h_{\max},T,\delta,\epsilon)$ returns a leaf of $\CT$ with $\cost$ no more than $\epsilon$ greater than the minimum leaf.}
Set $m = 0$.\;
\While{$\top$}
{
    Use $\qsubtree{\CA}{r,\branch,d,h_{\max},2^m,\delta/5\log{T}}$ to define oracle $\branch_1$.\;
    $\mathit{bound}_1 = \cost(\qtminleaf{\branch_1,\cost,d,c_{\max},4\cdot2^m,\delta/5\log{T}})$.\;
    Define $\cost_1$ as follows: for a node $n$, if $\branch_1(n)$ is empty but $\branch(n)$ is not empty, $\cost_1(n) = +\infty$. Otherwise $\cost_1(n) = \cost(n)$.\;
    $\mathit{incumbent} = \qtminleaf{\branch_1,\cost_1,d,c_{\max},4\cdot2^m,\delta/5\log{T}})$.\;
    Define $\branch_2$ as follows: for a node $n$, if $\branch_1(n) \ne \branch(n)$, $\branch_2(n)$ is empty. Otherwise $\branch_2(n) = \branch(n)$.\;
    $\mathit{bound}_2 = \cost(\qtminleaf{\branch_2,\cost,d,c_{\max},4\cdot2^m,\delta/5\log{T}})$.\;
    $\textit{best-bound} = \min(\mathit{bound}_1,\mathit{bound}_2)$. \;
    \If{$\cost(\mathit{incumbent}) \le \textit{best-bound} + \epsilon$} 
    {
        Return $\mathit{incumbent}$.\;
    }
}
\caption{\iqbbtext{}}
\label{algo:iqbb}
\end{algorithm}

\begin{proof}[\textbf{Proof of \thm{iqbb-main}.} ] We show that $\iqbb_{\CA}$ defined in \algo{iqbb} satisfies the conditions of \thm{iqts-main}. 
First, we ignore failure probabilities. Suppose after $Q(\epsilon)$ nodes that $\CA$ has explored enough nodes so that the gap is below $\epsilon$. Consider the loop iteration where $m_Q = \lceil \log(Q(\epsilon)) \rceil$. The sub-tree $\CT_1$ corresponding to $\branch_1$ contains at least the first $Q(\epsilon)$ nodes explored by $\CA$. \algo{iqbb} computes the $\incumbent$ as the minimum \emph{feasible} leaf of $\CT_{1}$ (as the cost of non feasible leaves of $\CT_1$ is set to $+\infty$ in $\cost_1$). Any discovered feasible solution in the first $Q(\epsilon)$ nodes is also in $\CT_1$ and so the estimate of $\incumbent$ has \emph{lower} cost than that found by $\CA$ after $Q(\epsilon)$ nodes.

The active nodes in $\CT_1$ fall into two groups:
\begin{enumerate}
    \item Nodes where none of their children is in $\CT_1$. These are leaves of $\CT_1$ and $\mathit{bound}_1$ finds the minimum $\cost$ among these.
    \item Nodes where only some of their children are in $\CT_1$. $\branch_2$ prunes $\CT_1$ to make these leaves and $\mathit{bound}_2$ computes the minimum among these and some active nodes of the first type.
\end{enumerate}
$\bestbound$ is therefore by construction the minimum $\cost$ of all the active nodes in $\CT_{1}$. The active nodes found by $\CA$ after exploring $Q(\epsilon)$ nodes are either interior nodes now (and the leaves descended from them have greater cost) or are still active nodes. The the computed $\bestbound$ is \emph{greater} than that found after $Q(\epsilon)$ nodes explored by $\CA$, which was $\epsilon$ by assumption.
The algorithm thus terminates after at most $m_Q = \lceil \log(Q(\epsilon) \rceil$ iterations.

At each iteration $m$, we make $\tilde{O}\left(d\sqrt{2^m}\log(h_{\max})\log^2\left(\frac{\log{T}}{\delta}\right)\right)$ queries for $\qsubtree{\mathcal{A}}{\cdot}$ and an additional $\tilde{O}\left(d\sqrt{2^m}\log(c_{\max})\log^2\left(\frac{\log{T}}{\delta}\right)\right)$ queries for $\qtminleaf{\cdot}$, leading to a total query complexity of $\tilde{O}\left(d\sqrt{2^m}\log(h_{\max}c_{\max})\log^2\left(\frac{\log{T}}{\delta}\right)\right)$. Across iterations, the total number of queries made is, therefore,
\begin{equation}
    \tilde{O}\left(\left(d\log(h_{\max}c_{\max})\log^2\left(\frac{\log{T}}{\delta}\right)\right)\sum_{m=0}^{m_Q} \sqrt{2^m}\right) = \tilde{O}\left(d\sqrt{2^{m_Q}}\log(h_{\max}c_{\max})\log^2\left(\frac{\log{T}}{\delta}\right)\right),
\end{equation}
where the equality follows from the formula for the sum of geometric series. Finally, note that there are at most $4m_Q$ possibilities for failure, each with probability $< \delta/5\log(T)$. Thus the total failure probability is at most $\delta$, and the proof is complete.
\end{proof}

\subsection{Universal Speedup for Branch-and-Cut}
\label{sec:iqbc}
\begin{algorithm}[!htbp]
\DontPrintSemicolon
\KwIn{$\branch$ and $\cost$ oracles of $\bnb{}$ tree $\CT$ of depth $d$ rooted at $r$, classical algorithm $\CA$ that explores the tree using heuristic $h$, upper bound $h_{\max}$ on heuristic, upper bound $T$ on size of tree,  upper bound $p$ on number of cutting planes discovered, failure probability $\delta$, precision parameter $\epsilon$, function $\cp$ that indicates if global cutting plane is found at a node}
\KwResult{$\iqbb_{\mathcal{A}}(r,\branch,\cost,\cp,d,c_{\max},h_{\max},p,T,\delta,\epsilon)$ returns a leaf of $\CT$ with $\cost$ no more than $\epsilon$ greater than the minimum leaf.}
Set $m = 0$.\;
Set $\branch',\cost' = \branch,\cost$.\;
\While{$\top$}
{
    Use $\qsubtree{\mathcal{A}}{r,\branch',d,h_{\max},2^m,\delta/4\log{T}}$ to define oracle $\branch_1$.\;
    $\mathit{bound}_1 = \cost'(\qtminleaf{\branch_1,\cost,d,c_{\max},4\cdot2^m,\delta/5\log{T}})$.\;
    Define $\cost_1$ as follows: for a node $n$, if $\branch_1(n)$ is empty but $\branch'(n)$ is not empty, $\cost_1(n) = +\infty$, otherwise $\cost_1(n) = \cost'(n)$.\;
    $\mathit{incumbent} = \qtminleaf{\branch_1,\cost_1,d,c_{\max},4\cdot2^m,\delta/5\log{T}})$.\;
    Define $\branch_2$ as follows: for a node $n$, if $\branch_1(n) \ne \branch(n)$, $\branch_2(n)$ is empty, otherwise $\branch_2(n) = \branch_1(n)$.\;
    $\mathit{bound}_2 = \cost(\qtminleaf{\branch_2,\cost,d,c_{\max},4\cdot2^m,\delta/5\log{T}})$.\;
    $\textit{best-bound} = \min(\mathit{bound}_1,\mathit{bound}_2)$. \;
    \If{$\cost(\mathit{incumbent}) \le \textit{best-bound} + \epsilon$} 
    {
        Return $\mathit{incumbent}$.\;
    }
    Use $\qtsearch{\branch_1,d, 4\cdot2^{m},\cp,\delta/5p\log{T}}$ $p$ times to find up to $p$ different global cutting planes.\;
    Redefine $\branch',\cost'$ to also include cutting planes in the definition of the internal relaxed problem.
}
\caption{Incremental Quantum \bnc}
\label{algo:iqbc}
\end{algorithm}

\begin{proof}[\textbf{Proof of \thm{iqbc-main}.}]
We focus on global cutting planes as local cutting planes can be trivially added to \algo{iqbb} (or the algorithm of Montanaro~\cite{montanaro2018backtracking}) by modifying the branch oracle so that cutting planes discovered at a node are added to its children). The proof of the query complexity and the failure probability follows easily from that of \thm{iqbb-main}. Suppose that the algorithm terminates after $Q(\epsilon)$ iterations, using $p'$ cutting planes that were discovered until that point. It is clear that after $\lceil \log(Q(\epsilon)) \rceil$ iterations, the $p'$ cutting planes will be discovered. These cutting planes are therefore available at iteration $\lceil \log(Q(\epsilon)) \rceil + 1$ and an approximately optimal solution is obtained (see the proof of \thm{iqbb-main}).
\end{proof}

\subsection{Partial Subtree Generation}
\label{sec:partial-subtree}
In this section we describe the proof of our main technical ingredient \thm{subtree-main}. Our first observation is that in order to describe a partial subtree generation procedure for a general Branch-Local heuristic it suffices to consider \emph{local} heuristics which rank active nodes by a function of only the data at the node. Specifically we have the following lemma

\begin{lemma}[Reduction to local heuristics]
\label{lem:heuristic-reduction}
Let $\CT$ be a tree rooted at $r$ and specified by a $\branch$ oracle, $\cost$ be an additional $\cost$ oracle on the nodes (can be chosen to always return a constant if no such oracle is inherent to the problem), and let $h$ be a Branch-Local heuristic that explores the tree by and ranks node $N$ on the basis of a function 
\begin{equation}
    \heur(N) = f(\hlocal(N),\hparent(r),\hparent(n_1)\dots,\hparent(n_{d(n)-1}),d(N))
\end{equation}
where $d(N)$ is the depth of the node $N$, $r \to n_1 \to \dots \to n_{d(N)-1} \to N$ is the path from $r \to N$, each of $\hlocal,\hparent$ make a constant number of queries to $\branch,\cost$, and $f$ is a function with no dependence on $N$. Then there is a new tree $\CT'$ such that
\begin{itemize}
    \item There exists a bijection $\Phi$ from the set of nodes of $\CT$ to $\CT'$.
    \item There exists a new oracle $\branch'$ that makes a constant number of queries to $\branch$ and $\cost$ and $\phi(N_2) \in \branch(\phi(N_1))$ if and only if $N_2 \in \branch(N_1)$.
    \item There exists a local function $\hcost$ that makes a constant number of queries to $\branch,\cost$ so that ranking nodes on the basis of $\hcost(\phi(N))$ is equivalent to ranking them on the basis of heuristic $h$.
\end{itemize}
\end{lemma}
\begin{proof}
The main idea is to allow the node to maintain a record of its depth and a transcript of the necessary information (value of $\hparent$) from its parents. Specifically $\phi(N)$ will consist of a tuple consisting of $N$, the depth $d(N)$, and a list of $\hparent(P)$ for every parent $P$ of $N$. We assign $\phi(r) = (r,0,\{\})$. We construct the oracle $\branch'$ as follows: add $(N_2,d+1,\mathbf{append}(l,\hparent(N_1))$ to $\branch'(N_1,d,l)$ if and only if $N_2 \in \branch(N_1)$. From the conditions on $\hparent$ (\defn{branch-local}), $\branch'$ makes only a constant number of queries to $\branch$.
By definition, $h$ ranks on the basis of 
\begin{equation}
\heur(N) = f(\hlocal(N),\hparent(r),\hparent(n_1)\dots,\hparent(n_{d(n)-1}),d(N)).
\end{equation} 
It follows from the construction above that there exists a function $\hcost$ such that $\hcost(\phi(N)) = \heur(N)$ for all nodes of $\CT$. Therefore exploring the nodes of $\CT'$ by ranking according to $\hcost$ (which is a function only of an individual node) is equivalent to exploring nodes in $\CT$ according to $h$. This completes the proof.
\end{proof}

\lem{heuristic-reduction} allows us to consider, without loss of generality, only those heuristics that rank nodes by some local function $\hcost$. In the rest of the paper we shall call this function the heuristic cost. It therefore suffices to prove \thm{subtree-main} only for such heuristics, we do so in \thm{subtree-main-local} to appear below.

We first make certain assumptions without loss of generality that will clarify the presentation of the proof.

\begin{assumption}
\label{assum:binary-tree}
All trees considered will be binary.
\end{assumption}
Suppose otherwise that the maximum degree were some constant $\degree$. We can replace a node and its $\degree$ children by a binary tree of depth at most $\log(\degree)$, ensuring that the new trees are binary. The asymptotic query complexities do not change.

\begin{assumption}
\label{assum:integer-cost}
The value of the cost function $\cost$ for every node $n$ in the \bnb{} tree is a positive integer.
\end{assumption}

\begin{assumption}
\label{assum:integer-hcost}
The value of the heuristic cost function $\hcost$ for every node $n$ in the \bnb{} tree is a positive integer.
\end{assumption}

These assumptions can be obtained simply by truncating the real value and scaling all the costs up to an integer. They may introduce a logarithmic dependence on the precision via the $\log(c_{\max}),\log(h_{\max})$ terms.

\begin{assumption}
\label{assum:heuristic-total-order}
The heuristic used by the algorithm does not rank any two active nodes exactly the same.
\end{assumption}
This assumption is reasonable in order for the classical algorithm to succeed. If it is not true, we can make the following modification. Suppose there are at most $A$ active nodes at a time. We simply toss $2\log(A)$ random coins and only declare two nodes equally ranked if both $\hcost$ and the random coins are the same. Due to the choice of number of coins, the probability of a collision is $1 - \omega(1)$ so no two nodes are ranked the same almost surely.

\subsubsection{Oracle Transforms}
Let $r_0$ be the root of $\CT$, and $n_0,n_1$ be its children. A primitive we will use repeatedly throughout the algorithm will be to truncate subtrees rooted at $r,n_0,n_1$ based on the value of heuristic cost $\hcost$. Below we define some oracle transforms that perform these truncations. In each case the new oracles will make a constant number of calls to the original tree oracle.

The first transform simply removes all nodes with heuristic cost above a threshold.
\begin{definition}[Truncation based on heuristic cost]
\label{defn:general-truncation}
Given a subtree of $\CT$ rooted at $r$ specified by a $\branch$ oracle, heuristic cost $\hcost$, and a threshold $t$, $\trunc{r,\branch,\hcost,t}$ returns an oracle $\branch'$ corresponding to a subtree of $\CT$ with all nodes with $\hcost$ greater than $t$ removed. Specifically, $\branch'(N) = \{x | x \in \branch(N), \hcost(N) <  t\}$. $\trunc{\cdot}$ makes a constant number of queries to $\branch$.
\end{definition}

The next transformation removes all nodes whose parent's heuristic cost is above a threshold.
\begin{definition}[Truncation based on heuristic cost of parent]
\label{defn:parent-truncation}
Given a subtree of $\CT$ rooted at $r$ specified by a $\branch$ oracle, heuristic cost $\hcost$, and a threshold $t$, $\trunc{r,\branch,\hcost,t}$ returns an oracle $\branch'$ corresponding to a subtree of $\CT$ such that all nodes with heuristic cost $> t$ whose parents have cost $\le t$ are leaves. Specifically, for any node $N$ with parent $N'$, $\branch'(N)$ is empty ($N$ is a leaf) if $\hcost(N') \le t < \hcost(N)$. $\ptrunc{\cdot}$ makes a constant number of queries to $\branch$.
\end{definition}

Finally, we define a transform that truncates the subtrees rooted at $n_0,n_1$ based on two different heuristic costs.
\begin{definition}[Truncate left and right subtrees]
\label{defn:two-sided-truncate}
Given a subtree of $\CT$ rooted at $r$ specified by a $\branch$ oracle, heuristic cost $\hcost$, and two thresholds $t_0,t_1$, $\twotrunc{r,\branch,\hcost,t_1,t_2}$ returns an oracle $\branch'$ corresponding to a subtree of $\CT$ such that all nodes in the subtree rooted at $n_0$ with heuristic cost greater than $t_0$, and all nodes in the subtree rooted at $n_1$ with heuristic cost greater than $t_1$ are removed. This can be accomplished by first defining a new oracle $\branch_1$ by adding a flag to each node indicating if it is descended from $n_0$ or $n_1$. Then $\twotrunc{}$ can be defined using a single call to $\branch_1$ which in turn uses one call to $\branch$.
\end{definition}

\subsubsection{Main Proof}
Before we state and prove \thm{subtree-main-local}, we present two primitives that we will employ
\begin{enumerate}
    \item A subroutine to find the $k^{\mathrm{th}}$ lowest heuristic cost in a subtree rooted at a child of $r$. This primitive is formally specified in \algo{kth-cost-node}, and satisfies the following guarantee
    \begin{lemma}
    $\kthcost{\branch,d,\hcost,h_{\max},k,\epsilon,\delta}$ (\algo{kth-cost-node}) makes $\tilde{O}\left(\frac{\sqrt{k}d}{\epsilon^{3/2}}\log\left(h_{\max}\right)\log^2\left(\frac{1}{\delta}\right)\right)$ queries to $\branch$ and $\hcost$, and returns with probability $\ge 1 - \delta$, such that the number of elements of $\{x | x = \hcost(N), N \in \CT'\}$ that are less than $c$ is strictly more than $k$ and less than or equal to $k(1+\epsilon)^{2}$ where $\CT'$ is the tree specified by $\branch$.
    \end{lemma}
    \begin{proof}
    To see the correctness, observe that
    \begin{itemize}
    \item For a cost $\alpha_{\min}$ such that $\le k$ nodes have $\hcost$ less than $\alpha_{\min}$: the $\qtsize{\cdot}$ call in \algo{kth-cost-node} returns an integer.
    \item For a cost $\alpha_{\max}$ such that $\ge k(1 + \epsilon_2)^2$ nodes have $\hcost$ less than $\alpha_{\max}$: the $\qtsize{\cdot}$ call in \algo{kth-cost-node} returns "more than $k$ nodes".
    \end{itemize}
    By \assum{heuristic-total-order}, $\alpha_{\min}$ and $\alpha_{\max}$ are different and any $\alpha$ in this window is an acceptable solution. The window has size at least $1$ and locating it in a domain of $[0,h_{\max}]$ requires $\tilde{O}(\log(h_{\max}))$ binary search iterations, leading to $\tilde{O}(\log(h_{\max}))$ queries to $\qtsize{\cdot}$, which yields the required query complexity.
    
    \begin{algorithm}[!htbp]
\DontPrintSemicolon
\KwIn{A subtree $\CT'$ of $\CT$ rooted at $r$ specified by a $\branch$ oracle, depth $d$, heuristic cost function $\hcost$, maximum heuristic cost $h_{\max}$, integer $k$, error parameter $\epsilon \le 1$, failure probability $\delta$}
\KwResult{$\kthcost{\branch,d,\hcost,h_{\max},k,\delta,\epsilon} = c$ such that $c$ is between the $k^{\mathrm{th}}$ and $k(1 + \epsilon_2)^{\mathrm{th}}$ element of the set $\{x | x = \hcost(N), N \in \CT'\}$ sorted in ascending order.}
Define $\delta' = \delta/\log\left({h_{\max}}\right)$.\;
Run binary search on the interval $[0,h_{\max}]$ to find the minimum $\alpha$ such that: 
\quad $\qtsize{\trunc{r,\branch,\hcost,\heur,\alpha},d,k,\delta',\epsilon}$ returns "more than  k nodes" \;
\KwRet $\alpha$.
\caption{Kthcost: Find $k^{\mathrm{th}}$  lowest heuristic cost in subtree.}
\label{algo:kth-cost-node}
\end{algorithm}
    
    \end{proof}
    \item A subroutine to find the minimum heuristic cost greater than or equal to a threshold $t$ in a subtree rooted at a child of $r$. This primitive is formally specified in \algo{next-cost}.
    \begin{lemma}
    $\nextcost{\branch,d,\hcost,h_{\max},c,\delta}$ makes $\tilde{O}\left(\sqrt{T_{\max}} d \log\left(h_{\max}\right) \log^2\left(\frac{1}{\delta}\right)\right)$ returns the minimum heuristic cost of any node in the subtree defined by $\branch$, among those nodes with heuristic cost greater than $c$.
    \end{lemma}
    \begin{proof}
    By construction $\branch',\hcost$ is a \bnb{} tree, and any node $n$ such that $\hcost(n) > c$ and $\hcost(p(n)) \le c$ where $p(n)$ is the parent of $n$, is marked as a leaf . Furthermore, $\cost'(n) = \hcost(n)$ for a leaf $n$. Thus the $\qtminleaf{\cdot}$ call returns the correct solution, and the runtime follows from \prim{qtminleaf}.
    \end{proof}
    
    \begin{algorithm}[!htbp]
\DontPrintSemicolon
\KwIn{A subtree $\CT'$ of $\CT$ rooted at $r$ specified by a $\branch$ oracles, depth $d$, heuristic cost function $\hcost$, maximum cost $h_{\max}$, upper bound $T_{\max}$ on number of nodes whose parents have cost at most $h_{\max}$, threshold $c'$ and failure probability $\delta$}
\KwResult{$\nextcost{\branch,d,\hcost,h_{\max},c,\delta}$, the minimum of the set $\{\hcost(N) | N \in \CT, \hcost(N) \ge c$\}.}
$\branch' = \ptrunc{r,\branch,\hcost,c}$.\;
Define $\cost'$ such that for any node $n \in \CT'$, $\cost'(n) = \hcost(n)$ if $\branch'(n)$ is empty, and $\cost'(n) = 0$ otherwise \; \label{lin:make-ptrunc-bnb}
Return $\qtminleaf{\branch',\cost',d,h_{\max},T_{\max},\delta}$.
\caption{MinimumNextCost: Find the lowest heuristic cost in a subtree, that is greater than given parameter $c$.}
\label{algo:next-cost}
\end{algorithm}
\end{enumerate}

\begin{algorithm}[!htbp]
\DontPrintSemicolon
\KwIn{A tree $\CT$ specified by $\cost$ and $\branch$ oracles, with root $r$, depth $d$, heuristic cost function $\hcost$, upper bound $h_{\max}$ on the $\hcost$ of any node, and failure probability $\delta$, integer input parameter $m$.}
\KwResult{$\qsubtreelocal{r,m,\branch,\hcost,d,\delta}$ returns $c_0,c_1$ such that $r \cup \trunc{n_0,\branch,\hcost,c_0} \cup \trunc{n_1,\branch,\hcost,c_1}$ is a sub-tree satisfying the conditions of \defn{subtree-estimation} }
Let $n_0,n_1$ be the children of root $r$ ordered such that $\hcost(n_0) \le \hcost(n_1)$. Let $t_0,t_1$ denote the trees rooted at $n_0,n_1$.\;
Set flags $\curr = 0$. \tcp*[r]{Indicates currently growing subtree: $\curr \in \{0,1\}$}
Set $T_0 = T_1 = 1$. \tcp*[r]{$T_1$ nodes currently included from $t_i$}
Set $m_0 = m_1 = 0$. \tcp*[r]{Attempt to include $2^{m_i}$ nodes from $t_1$ in next iteration}
Set $\epsilon_1, \epsilon_2 = \ln(2)/8$. \;
Set $\delta' = \delta/8(m+3)$. \;
Set $c_0 = 
\hcost(n_0), c_1 = \hcost(n_1))$. \;
Set $c'_0 = \max(\branch(n_0)), c'_1 = \max(\branch(n_1))$.\;
Set $\done = 0$.\;
Define $\branch_\curr  = \trunc{n_\curr,\branch,\hcost,c_{\neg \curr}}$\;
Define $\branch'_\curr  = \trunc{n_\curr,\branch,\hcost,c'_{\neg \curr}}$\;
 \While{$T_0  +  T_1  + 1 \le 2^{m+1}(1 + \epsilon_2)^{4}(1 + \epsilon_1)^{2}$} {\label{lin:size-check-cost}
    \While{$\qtsize{\branch_\curr,d,2^{m_\curr},\delta',\epsilon_2}$ returns ``more than $2^{m_\curr}$ nodes'' } {\label{lin:add-valid-nodes}
     \If{$2^{m_\curr} \le 2^{m+1}(1 + \epsilon_2)^{4}(1 + \epsilon_1)^{2} -  T_{\neg{\curr}} - 1$} {\label{lin:exit-check-2}
        Set $m_\curr \gets m_\curr  + 1$. \;\label{lin:increase-m-1}
        }
    \Else {
         Set $\done = 1$, exit loop.  \;\label{lin:set-done-1}
    }
    }
    \While{
    $c_{\neg \curr} \ne {c'}_{\neg \curr} \And \done = 0 \And \qtsize{\branch'_\curr,d,2^{m_\curr},\delta',\epsilon_2}$ returns ``more than $2^{m_\curr}$ nodes''
    } {\label{lin:look-forward-loop}
         \If{$2^{m_\curr} \le 2^{m+1}(1 + \epsilon_2)^{4}(1 + \epsilon_1)^{2} -  T_{\neg{\curr}} -1$ \label{lin:look-forward-if}} { \label{lin:exit-check-3}
        Set $m_\curr \gets m_\curr  + 1$. \;\label{lin:increase-m-2} 
        }
    \Else {
        Set $\done = 1$, exit loop.  \; \label{lin:set-done-2}
    }
     Set $c_\curr \gets \kthcost{n_\curr,\branch,\hcost,2^{m_\curr -1},\delta',\epsilon_2}$.\;\label{lin:intermediate-c-update}
     \If {$c_{\curr} < {c'}_{\neg{(\curr)}}$} {\label{lin:make-valid-again}
        Set $c_{\neg \curr} \gets \nextcost{n_{\neg \curr},\branch,\hcost,c_{\curr},\delta'}$.\;
        Set $T_{\neg \curr} \gets \qtsize{\trunc{n_{\neg \curr},\branch,\hcost,c_{\neg \curr}},d,2^{m_\curr - 1},\delta',\epsilon_1}$.\;
    }
    \Else{
        Exit loop.\;
    }
    
    }
    \If{$\done = 1$} {
        Set $c_\curr \gets \kthcost{n_\curr,\branch,\hcost,2^{m+1}(1 + \epsilon_2)^{4}(1 + \epsilon_1)^{2} -  T_{\neg{\curr}} -1,\delta', \epsilon_1}$.\;\label{lin:kthcost-done=1}
         Exit loop.\;\label{lin:done-termination}
    }
    \Else {
        Set $c_\curr \gets \nextcost{n_\curr,\branch,\hcost,c'_{\neg \curr},\delta'}$.\;\label{lin:set-c-not-done}
        Set $c'_\curr \gets \kthcost{n_\curr,\branch,\hcost,2^{m_\curr},\delta',\epsilon_2}$.\;\label{lin:set-c-prime}
    }
    For $i \in \{0,1\}$, $T_i \gets \qtsize{\trunc{n_i,\branch,\hcost,c_i},d, 2^{m_i - 1},\delta',\epsilon_1}$. \;
    Set $\curr \gets \neg (\curr)$.\;
}
\KwRet $c_0,c_1$\;
\label{lin:final-return}
\caption{Quantum Partial Subtree Generation}
\label{algo:cost-based-subtree}
\end{algorithm}

\begin{theorem}
\label{thm:subtree-main-local}
Let $\CT$ be a tree with root $r$, specified by a $\branch$ oracle. Given a classical algorithm $\CA$ that uses a heuristic that ranks by a function $\hcost$ of the nodes , $\qsubtreelocal{r,m,\branch,\hcost,d,h_{\max},\delta}$ (see \algo{cost-based-subtree}) uses $\tilde{O}\left(d\sqrt{2^{m}}\log(h_{\max})\log^2\left(\frac{1}{\delta}\right)\right)$ queries to $\branch,\hcost$ and returns with probability at least $1-\delta$, returns $c_0,c_1$ such that the tree corresponding to $r \cup \trunc{n_0,\branch,\hcost,c_0} \cup \trunc{n_1,\branch,\hcost,c_1}$ contains at most $4\cdot2^m$ nodes and at least the first $2^{m}$ nodes visited by $\CA$. 
\end{theorem}

\paragraph{Outline of Proof.} We denote the children of the root $r$ by $n_0,n_1$ so that $n_0$ has lower heuristic cost, and denote the sub-tree rooted at $n_i$ by $t_i$. In the following, we shall use the phrase \emph{returned subtree} to refer to $r \cup \trunc{n_0,\branch,\hcost,c_0} \cup \trunc{n_1,\branch,\hcost,c_1}$. It is clear from \defn{general-truncation} that the returned subtree contains the root $r$, and any nodes from $t_i$ with heuristic cost less than $c_i$. At any moment, $\curr$ is a flag that indicates that nodes are being added to $t_\curr$ by \algo{cost-based-subtree}. We maintain $c_i$ such that, when nodes from $t_\curr$ are being added to the returned subtree, $c_{\neg \curr}$ represents the lowest heuristic cost of an 'unexplored' node from $t_{\neg \curr}$. We can therefore add any nodes in $t_\curr$ with cost less than $c_{\neg \curr}$ to the returned subtree, as the classical algorithm also explores those nodes first. To verify that more than $4 \cdot 2^{m}$ nodes are not added, we propose to add $2^{m_\curr}$ nodes from $t_\curr$ to the returned subtree, and  use $\qtsize{\cdot}$ to check that there indeed exist $2^{m_\curr}$ nodes in $t_\curr$ with heuristic cost less than $c_{\neg \curr}$. If there are a sufficient number of nodes, we add $2^{m_\curr}$ nodes to the returned subtree, and increment $m_{\curr}$. Otherwise $c_{\curr}$ is updated accordingly, and $\curr$ is negated, i.e. we start adding nodes from the other subtree. Note that it is crucial that we propose to add an exponentially increasing number of nodes, i.e., if there are a constant number of queries at a single value of $m_i$, the total oracle complexity of the calls to $\qtsize{\cdot}$ is asymptotically the same as that with the greatest value of $m_i$ (due to the behavior of sums of geometric series). The scheme just described suffers from the following problem: it could be that the sequence of nodes explored by the classical algorithm oscillates back and forth between the $t_0$ and $t_1$, adding a small number of nodes at a time without doubling the number of nodes added from either subtree. Then, our proposed scheme would repeatedly switch between $t_0$ and $t_1$ while failing to add $2^{m_0}$ or $2^{m_1}$ nodes to the respective subtree. This could occur $\Theta(2^m)$ times, leading to $\Theta(2^m)$ calls to $\qtsize{\cdot}$, and the oracle complexity would be larger than desired. To ensure that the oracle complexity is $O(\sqrt{2^m})$ we must ensure that, for any value of $m_\curr$, the attempt to add $2^{m_\curr}$ nodes fails at most a constant number of times. We therefore maintain a new set of values $c'_0,c'_1$, that are assigned as follows: suppose we have added $2^{m_i-1}$ nodes to $t_i$ and $2^{m_i}$ nodes cannot be added. Then $c'_i$ is a value chosen such that $2^{m_i}$ nodes with heuristic cost less than $c'_i$. In principle therefore, there is no need to include nodes from $t_i$ until nodes with cost less than $c'_i$ from $t_{\neg i}$ have been included, and we could use the original scheme with $c'_{i}$ taking the place of $c_i$. We note however, that in this case, we add nodes in a different order from the classical algorithm. In particular any node from $t_\curr$ with cost greater than $c_{\neg \curr}$ and less than $c'_{\neg \curr}$ would not be explored by the classical algorithm until some other nodes from the other tree are explored. In the worst case, there could be enough nodes with heuristic cost between $c_{\curr}$ and $c_{\neg \curr}$ so that no more nodes from $t_{\neg curr}$ are added, which makes the algorithm incorrect. To resolve this problem, we add nodes between $c_{\curr}$ and $c_{\neg \curr}$ \emph{incrementally}, ensuring that the total number of nodes added is never more than double of what would be explored by the classical algorithm. The detailed proof to follow, describes the details of this procedure and show that it works, establishes that $c,c'$ are assigned appropriately. A final detail is that $\qtsize{\cdot}$ can be evaluated only to constant relative error, and we must carefully control the accumulation of these errors throughout the procedure.

\begin{proof}[{\textbf{Proof of \thm{subtree-main-local}}}]
 We first show that the algorithm is correct if we ignore the failure probabilities of the subroutines. Correctness with the appropriate failure probability will then follow from the union bound. 
\paragraph{Correctness.} 
We first estimate the size of the trees that may be returned by the algorithm.  By definition from Line~{\ref{lin:final-return}} of \algo{cost-based-subtree}, the returned subtree contains the root $r$, and any nodes from $t_i$ with cost less than $c_i$. Define $\hat{T}_i$ to be the true number of nodes returned from $t_i$, i.e., there are exactly $\hat{T}_{i}$ nodes in $t_i$ with heuristic cost less than or equal to $c_i$. The returned tree contains $\hat{T}_0 + \hat{T}_1 + 1$ nodes. We also observe that $T_i$ is an estimate of $\hat{T}_i$ obtained by $\qtsize{\cdot}$ with relative error $\epsilon_1$. Therefore,
\begin{equation}
\label{eq:T-estimates}
    \forall i \in \{0,1\}, \qquad \hat{T}_i \le T_i \le (1 + \epsilon_1)^2\hat{T}_i
\end{equation}

We consider the conditions that may lead to exiting the outermost loop (Line~\ref{lin:size-check-cost}).
\begin{itemize}
    \item Flag $\done$ is set to 1. There are two possibilities (Line~\ref{lin:set-done-1} and Line~\ref{lin:set-done-2}) for the first time this happens. In either case, the algorithm terminates, and $c_\curr$ is set to the minimum cost that ensures that there are between $2^{m+1}(1 + \epsilon_2)^4(1 + \epsilon_1)^2 - T_{\neg \curr} - 1$ and $2^{m+1}(1 + \epsilon_2)^4(1 + \epsilon_1)^4 - T_{\neg \curr} - 1$ nodes in $t_{\curr}$ with cost less than $c_\curr$ (due to the guarantee on $\kthcost{\cdot}$, Line \ref{lin:kthcost-done=1}). Finally from \eq{T-estimates}, the returned tree contains at least $2^{m+1}(1 + \epsilon_2)^4$ and at most $2^{m+1}(1 + \epsilon_2)^4(1 + \epsilon_1)^4$ nodes.
    \item $T_0  +  T_1  + 1 \le 2^{m+1}(1 + \epsilon_2)^{4}(1 + \epsilon_1)^{2}$ and flag $\done$ is set to  0.  The guards in (Line~\ref{lin:exit-check-2} and Line~\ref{lin:exit-check-3}) ensure that since flag $\done=0$, $c_\curr$ is chosen so that $2^{m_\curr-1} \le \hat{T}_\curr \le 2^{m_\curr-1}(1 + \epsilon_2)^2$. Since $\done = 0$, we also have $2^{m_\curr -1} \le 2^{m+1}(1+\epsilon_2)^4(1 + \epsilon_1)^2 - T_{\neg \curr} - 1$.
    This guarantees that $T_{\neg \curr} + 2^{m+1} + 1 < 2^{m+2}(1 + \epsilon_2)^{4}(1 + \epsilon_1)^{2}$. From \eq{T-estimates}, the returned tree contains at most $2^{m+1}(1 + \epsilon_2)^4$ and at most $2^{m+1}(1 + \epsilon_2)^4(1 + \epsilon_1)^4$ nodes. Furthermore, $T_0  +  T_1  + 1 \le 2^{m+1}(1 + \epsilon_2)^{4}(1 + \epsilon_1)^{2}$ implies from \eq{T-estimates} that the returned tree contains at least $2^{m+1}(1 + \epsilon_2)^4$ nodes.
\end{itemize}

Observe that the algorithm returns at least $2^{m+1}(1 + \epsilon_2)^4$ nodes. Given any sequence of included nodes we now define a notion of \emph{validity} as follows.
  \begin{definition}
\label{defn:node-validity}
Consider a \bnb{} tree explored by a classical algorithm $\CA$, and a set of nodes $\CS$.  We call a node $n \in \CS$ \emph{valid} if and only if every node that $\CA$ visits before $n$ is included in $\CS$.
\end{definition}
 Suppose the returned list does not contain the first $2^{m}$-nodes explored by the classical algorithm. Then  any node not in the first $2^m$ explored is then necessarily invalid by \defn{node-validity}, and the fraction of valid nodes is strictly less than $2^{m}/2^{m+1}(1 + \epsilon_2)^4 = 1/2(1 + \epsilon_2)^4$. If it is shown that at least $1/2(1 + \epsilon_2)^{4}$ fraction of the nodes are valid, the returned tree contains at least the first  $2^{m}$ nodes explored by the classical algorithm.

 Specifically we wish to establish the following proposition
 \begin{proposition}
 \label{prop:main-invariant-cost} During the execution of \algo{cost-based-subtree}, $c_0,c_1$ are such that at least $1/2(1 + \epsilon_2)^{4}$ fraction of nodes in the set $r \cup \{n \in t_0| \hcost(n) < c_0\} \cup \{n \in t_1| \hcost(n) < c_1\}$ are valid.
 \end{proposition}
 \prop{main-invariant-cost} can be proved through induction as follows. At the start of the algorithm $c_0$ and  $c_1$ are both chosen such that no nodes from $t_0,t_1$ are returned. Therefore all the returned nodes are valid as any search algorithm explores the root of the tree first. 
 We split the execution of the algorithm into phases based on the outermost loop (defined at Line~\ref{lin:size-check-cost}) as $c_0,c_1$ are only updated within the loop body. The loop body has two exit conditions:
 \begin{itemize}
     \item Return to loop guard at Line~\ref{lin:size-check-cost}. We denote the program state after the $i^{th}$ execution of this line by $L_i$.
     \item Loop break at Line~\ref{lin:done-termination}, when flag $\done  = 1$. Suppose this line is executed after $i$ executions of the loop guard (Line~\ref{lin:size-check-cost}). We denote the program state before that execution as $P_i$.
 \end{itemize}
 Correspondingly we define two invariants:
 \begin{itemize}
     \item $\CI_1(i)$: In $L_i$, the following conditions are satisfied:
     \begin{enumerate}
         \item $c_0, c_1$ are such that at most $1/2(1 + \epsilon_2)^4$ fraction of the returned nodes are valid.
         \item All invalid nodes are in $t_{\neg \curr}$.
         \item The number of nodes in $t_{\neg \curr}$ with heuristic cost $c < c'_{\neg \curr}$ is at most $2(1+ \epsilon_2)^{2}$ times those with heuristic cost $c < c_{\neg \curr}$.
     \end{enumerate}
     \item $\CI_2(i)$: In $P_{i}$ $c_0,c_1$ are such that at least $1/2(1+ \epsilon_2)^4$ fraction of returned nodes are valid.
 \end{itemize}
 If $\CI_1,\CI_2$ are maintained for all $i$, \prop{main-invariant-cost} holds. 
 
  Finally, we notice that due to the form of the return statement, for any node $n \in t_i$, any other node $n' \in t_i$ explored before $n$ is also included. Therefore any nodes that make $n$ \emph{invalid} must have been included from $t_{\neg i}$. We prove by induction that $\CI_1,\CI_2$ are maintained for all $i$. It is clear from the above discussion and the initialization of algorithm variables, that $\CI_1(0)$ holds, which proves the base case of the induction. We show the following two inductive steps: 
  \begin{itemize}
      \item $\CI_1(i) \implies \CI_1(i+1)$ : \\
      If the $\done$ flag is set to 1 during the execution of the loop (Line~\ref{lin:size-check-cost}), the algorithm will terminate and $L_{i+1}$ will never be reached, making $\CI_1(i+1)$ vacuously true. Consider now the case when $\done$ is never set to 1. The first execution will be of the loop at Line~\ref{lin:add-valid-nodes} and since $\done$ is not set to 1, the loop terminates due to its guard failing, i.e., $\qtsize{\branch_\curr,\cost_\curr,2^{m_\curr},\delta',\epsilon_2}$ does not return ``more than $2^{m_\curr}$ nodes''. At this point $m_\curr$ is such that there are less than $2^{m_\curr}(1 + \epsilon_2)^2$ nodes with heuristic cost less than $c_{\neg \curr}$ (or the guard would succeed) and more than $2^{m_\curr - 1}$ such nodes (or the guard would fail in the previous iteration). Since at least $2^{m_\curr - 1}$ nodes will be included from $t_\curr$, all the invalid nodes must now belong to $t_\curr$. 
      
      If the loop at Line~\ref{lin:look-forward-loop} is not executed, $c_{\neg \curr} = c'_{\neg \curr}$. Then, we find that $c_\curr$ is chosen such that every node with heuristic cost $\le c_{\neg \curr}$ has heuristic cost $< c_\curr$ (due to the guarantee on $\nextcost{\cdot}$). This makes all the returned nodes valid in $L_{i+1}$. Furthermore, due to the guarantees on $\kthcost{\cdot}$, at most $2^{m_\curr}(1 + \epsilon_2)^2$ nodes have cost less than $c'_{\curr}$. As at least $2^{m_\curr - 1}$ nodes have cost less than $c_\curr$, the number of nodes in $t_{ \curr}$ with heuristic cost $c < c'_{\curr}$ is at most $2(1+ \epsilon_2)^{2}$ times those with heuristic cost $c < c_{\curr}$.
      
      Otherwise, consider an execution of the loop at Line~\ref{lin:look-forward-loop}. It is possible that this loop is never entered, in which case, we reduce to the above.
      
      Otherwise, let the final clause in the guard succeed at some point, when $m_\curr = \alpha$. Then there are at least $2^{\alpha}$ nodes in $t_i$ less than $c'_{\neg \curr}$ and $m_\curr$ is set to $\alpha + 1$. Line~\ref{lin:intermediate-c-update} sets $c_\curr$ to a value such that between $2^{\alpha}$ and $2^{\alpha}(1 + \epsilon_2)^2$ nodes in $t_\curr$ have cost less than $c_\curr$. Note that success in the guard at Line~\ref{lin:make-valid-again} makes all nodes valid without changing the value of $c_\curr,c_{\neg \curr}$. Therefore, let the guard fail for the first time when $m_\curr = \alpha$ at the start of the loop. We therefore have at most $2^{\alpha + 1}(1 + \epsilon_2)^{2}$ nodes with heuristic cost less than $c'_{\neg \curr}$. Since the guard succeeded in the previous iteration, $c_{\neg \curr}$ must have been set to a value so that at least $2^{\alpha}$ nodes with heuristic cost less than $c'_{\neg \curr}$.
      Again, $c_\curr$ is chosen such that every node with heuristic cost $\le c_{\neg \curr}$ has heuristic cost $< c_\curr$. It follows from the above that at least $1/2(1 + \epsilon_2)^{2}$ fraction of the nodes in $t_{\curr}$ with cost less than $c_{\curr}$ are valid. Since the guard in Line~\ref{lin:make-valid-again} did not fail in an earlier iteration, there are at least $2^{\alpha}$ nodes in $t_i$ with heuristic cost less than $c_{\curr}$ and $c'_\curr$ is set so that at most $2^{\alpha + 1}$ nodes in $t_i$ have heuristic cost less than $c_{\neg \curr}$. Therefore, the number of nodes in $t_{ \curr}$ with heuristic cost $c < c'_{\curr}$ is at most $2(1+ \epsilon_2)^{2}$ times those with heuristic cost $c < c_{\curr}$.
      
      The last line before $L_i$ sets $\curr$ to $\neg \curr$ and from the above cases, $\CI_1(i + 1)$ is true.
      
      \item $\CI_1(i) \implies \CI_2(i+1)$ : \\
      If $P_{i+1}$ is reached, $\done$ must be set to $1$. This can happen at two points in the loop body. If $\done$ is set to 1 at Line~\ref{lin:set-done-1}, the guard at Line~\ref{lin:exit-check-2} must fail, while the guard at Line~\ref{lin:add-valid-nodes} must succeed. Therefore there are at least $2^{m_\curr}$ nodes with heuristic cost $< c_{\neg \curr}$, and $2^{m_\curr} > 2^{m+1}(1 + \epsilon_2)^{4}(1 + \epsilon_1)^{2} -  T_{\neg{\curr}} - 1$.  When $\done = 1$, $c_{\curr}$ is chosen so that the number of returned nodes from $t_\curr$ is at most $(1 + \epsilon_2)^{4}(2^{m+1}(1 + \epsilon_2)^{4}(1 + \epsilon_1)^{2} -  T_{\neg{\curr}} -1)$. At least $1/(1 + \epsilon_2)^2$ fraction of these nodes are valid, and from $\CI_1(i)$ at least $1/(1 + \epsilon_2)^2$ fraction of the returned nodes from $t_{\neg \curr}$ are valid. Thus $\CI_2(i+1)$ is true.
      
      The other possibility is that $\done$ is set to 1 in Line~\ref{lin:set-done-2}, therefore the guard at Line~\ref{lin:exit-check-3} fails while that at Line~\ref{lin:look-forward-if} succeeds. Then there are at least $2^{m_\curr}$ nodes with heuristic cost $< c'_{\neg \curr}$, and $2^{m_\curr} > 2^{m+1}(1 + \epsilon_2)^{4}(1 + \epsilon_1)^{2} -  T_{\neg{\curr}} - 1$.  When $\done = 1$, $c_{\curr}$ is chosen so that the number of returned nodes from $t_\curr$ is at most $(1 + \epsilon_2)^{4}(2^{m+1}(1 + \epsilon_2)^{4}(1 + \epsilon_1)^{2} -  T_{\neg{\curr}} -1)$ (from Line~\ref{lin:kthcost-done=1}, and $\epsilon_1 = \epsilon_2$ by definition). At least $1/(1 + \epsilon_2)^2$ fraction of these nodes have heuristic cost $< c'_{\neg \curr}$, and from $\CI_1(i)$, at least $1/2(1 + \epsilon_2)^2$ fraction of those are valid. As a whole, $1/2(1 + \epsilon_2)^4$ fraction of nodes in $t_\curr$ with cost less than $c_\curr$ are valid. Let $\beta$ be the value of $c_{\neg \curr}$ before the first iteration of the loop at Line~\ref{lin:look-forward-loop}. Since $\done$ was not set to $1$ in Line~\ref{lin:set-done-1}, all the nodes with cost less than $\beta$ must remain valid throughout. Furthermore, due to the guard at Line~\ref{lin:make-valid-again}, $c_{\neg \curr}$ is never incremented beyond $c'_{\neg \curr}$. Thus the fraction of valid nodes in $t_{\neg \curr}$ is at least $1/2(1 + \epsilon_2)^2$, from $\CI_1(i)$. Thus $\CI_2(i+1)$ holds.
  \end{itemize}
  We have shown that $\CI_1(1)$ is true, and that $\CI(i) \implies \CI_1(i+1) \bigwedge \CI_2(i+1)$. By induction therefore, both invariants hold throughout and at least $1/2(1 + \epsilon_2)^2$ fraction of returned nodes are valid. Thus the algorithm does return the first $2^m$ nodes explored by the classical algorithm. The algorithm returns at most $2^{m+1}(1 + \epsilon_2)^4(1 + \epsilon_1)^4$ nodes. Choosing $\epsilon_1 = \epsilon_2 = \ln(2)/8$ ensures that the number of returned nodes is at most $4\cdot2^{m}$. Thus if none of the subroutines fail, correctness is proved.
  
  \paragraph{Query Complexity.}
  When $\done=1$, the algorithm terminates immediately. We thus analyze the oracle complexity in the case where $\done$ remains 0 throughout, until the guard  Line~\ref{lin:size-check-cost} fails.  Consider an execution of the outermost loop. From the assignment of $c_{\curr}$ (Line~\ref{lin:set-c-not-done}), there are more nodes in $t_{\curr}$ with cost less than $c_{\curr}$ than with cost less than $c'_{\neg \curr}$. Furthermore, it follows from Line~\ref{lin:look-forward-loop} that there are at least $2^{m_\curr - 1}$ nodes with cost less than $c'_{\neg \curr}$. From \prim{qtsize}, we therefore have that $T_\curr > 2^{m_\curr - 1}$. As a consequence if $2^{m_\curr} > 2^{m+3}$, we have $T_0 + T_1 + 1 \ge 2^{m+2} > 2^{m+1}(1 + \epsilon_2)^4(1+\epsilon_1)^2$, leading to termination. We therefore have that $m_0 \le m+3, m_1 \le m+3$. 
  
  Notice that each iteration of loops at Line~\ref{lin:add-valid-nodes}, Line~\ref{lin:look-forward-loop} increases $m_\curr$ by 1.
  We  also notice that if $\done=0$, $c'_0,c'_1$ are set at initialization and subsequently in Line~\ref{lin:set-c-prime} is such that at least $2^{m_i}$ nodes in $t_i$ have cost less than $c'_i$ for $i \in \{0,1\}$. 
  As a consequence, in each iteration of the outermost loop $m_\curr$ is incremented at least once in either Line~\ref{lin:increase-m-1} or Line~\ref{lin:increase-m-2}. By the initialization of $c'_0,c'_1$, $c_\curr \ne c_{\neg \curr}$ in the first two iterations of the outermost loop, and $m_{\curr}$ is incremented in Line~\ref{lin:increase-m-2}. Otherwise, it must be that in the last iteration of the loop $c_{\neg \curr}$ was set to be greater than $c'_{\neg \curr}$. At the same time, it is true that two iterations ago $c'_{\neg \curr}$ was chosen so that at least $2^{m_\curr}$ nodes have heuristic cost less than $c'_{\neg \curr}$. Therefore, either $m_{\curr}$ is incremented in Line~\ref{lin:increase-m-1}, and otherwise $c_{\curr} \ne {c_{\neg\curr}}$ and $m_\curr$ is incremented in Line~\ref{lin:increase-m-2}. Since $m_0, m_1 \le m + 3$, there can be no infinite loops and the algorithm terminates.
  
  To bound the oracle complexity we argue as follows: the calls to $\branch$ and $\hcost$ are through calls to the procedures $\qtsize{\cdot}$, $\kthcost{\cdot}$, and $\nextcost{\cdot}$. The query complexity of each of these calls is $\tilde{O}\left(\frac{\sqrt{T_\mathrm{size}}d}{\epsilon^{3/2}}\log(h_{\max})\log^{2}\left(\frac{1}{\delta}\right)\right)$ where $T_\mathrm{size} = O(\sqrt{2^{m_\curr}})$. Now we define three lists $a,b,c$ such that the $i^{\mathrm{th}}$ element (which we denote $a_i,b_i,c_i$) of each list is the value of $m_{\curr}$ at the $i^{\mathrm{th}}$ call of $\qtsize{\cdot}$, $\kthcost{\cdot}$, and $\nextcost{\cdot}$ respectively. We now notice that no more than a constant number of the elements of $a$ can be equal. This is because loops at Line~\ref{lin:add-valid-nodes}, and Line~\ref{lin:look-forward-loop} increment $m_{\curr}$, so the only way to have more than a constant number of elements of $a$ with equal value is for the loop guards to fail more $\omega(1)$ times. From the argument above however, this is impossible, and in each iteration of the outermost loop, $m_{\curr}$ is incremented at least once in either Line~\ref{lin:increase-m-1} or Line~\ref{lin:increase-m-2}. By an identical argument, no more than a constant number of the elements of $b, c$ can be equal. Finally, as $a_i$ is a list of values each less than $2^{m+3}$, of which no more than a constant number are equal $\sum_i \sqrt{2^{a_i}} \le \sum_{j=0}^{m+3} \sqrt{2^j} = O(\sqrt{2^{m}})$. Similarly, $\sum_i \sqrt{2^{b_i}} = O(\sqrt{2^{m}})$ and $\sum_i \sqrt{2^{c_i}} = O(\sqrt{2^{m}})$. The total query complexity is given by $\tilde{O}\left(\frac{\left(\sum_i \sqrt{2^{a_i}} + \sum_i \sqrt{2^{b_i}} + \sum_i \sqrt{2^{c_i}}\right)d}{\epsilon^{3/2}}\log(h_{\max})\log^{2}\left(\frac{1}{\delta}\right)\right) = \tilde{O}\left(\frac{\sqrt{2^m}d}{\epsilon^{3/2}}\log(h_{\max})\log^{2}\left(\frac{1}{\delta}\right)\right)$.

  \paragraph{Controlling failure probability.} By the arguments above, there are fewer than $4(m + 3)$ executions of $\qtsize{\cdot}$, fewer than $2(m+3)$ executions of $\kthcost{\cdot}$, and fewer than $2(m+3)$ executions of $\nextcost{\cdot}$. By the union bound, choosing the failure probability of each procedure to be less than or equal to $\delta' = \delta/8(m+3)$ suffices to lower-bound the probability that any one procedure fails by $\delta$.
\end{proof}

\thm{subtree-main-local} allows us to prove our main result (\thm{subtree-main}) on Partial Subtree Estimation.

\begin{proof}[{\textbf{Proof of \thm{subtree-main}}}]
By \lem{heuristic-reduction} it suffices to consider a heuristic that ranks on the basis of some heuristic cost $\hcost$. We set $\qsubtree{\CA}{r,\branch,d,h_{\max},\delta} = \qsubtreelocal{r,\branch,\hcost,d,h_{\max},\delta}$. $\hcost$ makes only a constant number of queries to $\branch,\cost$ on the original tree, the oracle complexity of thus defined $\qsubtree{\CA}{\cdot}$ follows from \thm{subtree-main-local}. To show that the conditions of $\thm{subtree-main}$ are satisfied, we note the following:
\begin{itemize}
    \item $c_0,c_1$ can be used to define an oracle $\branch'$ corresponding to the tree consisting of the root $r$, $\trunc{n_0,\branch,\hcost,c_0}$, and $\trunc{n_1,\branch,\hcost,c_1}$, using the oracle transform $\twotrunc{r,\branch,c_0,c_1}$.
    \item By \thm{subtree-main-local}, $r \cup \trunc{n_0,\branch,\hcost,c_0} \cup \trunc{n_1,\branch,\hcost,c_1}$ contains at most $4\cdot2^m$ nodes, including at least the first $2^m$ visited by $\hcost$.
    \item By \defn{two-sided-truncate}, the resulting  $\branch'$ oracle makes a constant number of queries to $\branch$.
\end{itemize}
\end{proof}

\section{Numerical Study}
\label{sec:experiments}

In order to quantify the speedup that can be expected by using the proposed Quantum \bnb{} algorithm, we evaluate the performance of widely used classical solvers, Gurobi and CPLEX ~\cite{gurobi,cplex}, on a series of optimization problems. As discussed in section \sec{classical-bnc}, two main factors determine the complexity of the proposed quantum algorithm: the number of nodes explored and the depth of the tree. To quantify the projected performance of \iqbbtext{}, we perform numerical experiments for three problems that are widely studied in the literature: Sherrington-Kirkpatrick (SK) model\cite{sherrington1975solvable}, the maximum independent set (MIS) \cite{Karp1972} and portfolio optimization \cite{cornuejols2006optimization}. We set the hyperparameters of the classical optimizers to use up to 8 threads and to deactivate pruning and feasibility heuristics---computation time is only spent in \bnb{}---and the cutting planes but for the root node. Our goal is to answer the following four questions:

\begin{enumerate}
    \item \emph{What is the typical dependence of the number of nodes explored on the problem size parameters?} We numerically simulate many random instances of each problem type at different problem sizes, and record the number of nodes explored by Gurobi and CPLEX. The median number of nodes over these random instances at each problem size is treated as the typical behavior and we use curve fitting to determine the size dependence. For all investigated settings the dependence is observed to be exponential, i.e., $Q \sim 2^{\alpha n}$ where $n$ is the problems size, and $\alpha$ is a constant. The exponential $\alpha$ is inferred from regression. The obtained data and regression quality is discussed in \sec{numerics-sk}, \sec{numerics-mis}, and \sec{numerics-portfolio}.
    \item \emph{How large is the depth compared to the number of nodes?} Given the form of our quantum speedup ~(\thm{iqbb-main}), a necessary condition for near-quadratic advantage is that $d \ll Q$. To verify this for each of our three problem types, we record the maximum depth reported by Gurobi.\footnote{For a variety of reasons including the fact that node depth is only reported at constant time increments, and based on comments from Gurobi support, the node depth reported by Gurobi may not be an exact estimate of the true maximum depth of any node explored. Nonetheless this metric is the closest that can be extracted from the solver log, and we assume that it is a valid approximation to the true depth, at least in terms of its dependence on $n$. For more information see \href{https://support.gurobi.com/hc/en-us/community/posts/5928495659921-How-does-gurobi-calculate-the-Node-Depth-number-shown-in-logs-}{here}.} We plot (see \fig{depth_dependency}) the ratio of this maximum depth $d_{\max}$ to the square $n^2$ of the problem size, as a function of increasing problem size. We observe from these plots that $d_{\max}/n^2$ is a non-increasing function, indicating that $d = O(n^2)$. From the exponential dependence of $Q$ on $n$, it therefore follows that $d = o(Q)$.
    \item \emph{What is the inferred performance of our quantum algorithm?} Our observations above allow us to use our theoretical results to project the performance of a quantum algorithm obtaining a universal speedup over classical solvers. First we note that the work performed to implement $\branch$ and $\cost$ usually involves solving a continuous and convex optimization problem and is $\poly(n)$ (negligible compared to the exponential number of nodes explored). This work is also no more for the quantum algorithm than the classical algorithm~(in some cases quantum algorithms for continuous optimization~\cite{brandao2017sdp_better,kerenidis2018sdp,kerenidis2018sdp,van2020convex,chakrabarti2020quantum} can be used to accelerate these procedures). We benchmark the performance of the classical \bnb{} procedures by simply reporting the typical number $Q$ of nodes explored. We observed earlier that the depth $d$ is a logarithmic function of $Q$ and we report the quantum complexity as $\tilde{O}(\sqrt{Q})$. Our results are summarized in Table~\ref{table:numerics}.
    \item \emph{What is the spread in performance of \bnb{} over different instances of a fixed problem size?} A final consideration is the spread of the nodes explored $Q$ over the random instances at each problem size. This is interesting for two reasons: firstly, it provides a measure of how accurately our extrapolation measures the performance. Secondly, large spreads in performance further highlight the value of universal speedups that allow our quantum algorithm to leverage the structures that allow classical algorithms to perform much better than worst case for some instances. We estimate the spread as follows, for the largest problem size considered we calculate the spread in performance as the difference between the maximum and minimum number of nodes explored, which is then reported as a percentage of the median, see Table \ref{table:spread}.
\end{enumerate}

\newcommand{\SKgurobi}{$O(2^{0.494n})$}
\newcommand{\SKcplex}{$O(2^{0.513n})$}
\newcommand{\MISgurobi}{$O(2^{0.0373n})$}
\newcommand{\MIScplex}{$O(2^{0.0339n})$}
\newcommand{\PORTgurobi}{$O(2^{0.195n})$}
\newcommand{\PORTcplex}{$O(2^{0.140n})$}

\begin{table}[h]
    \centering
    \begin{tabular}{|c||c|c|c|c|}
    \hline
    \multicolumn{1}{|c||}{} & \multicolumn{2}{c|}{\textbf{Gurobi}} & \multicolumn{2}{c|}{\textbf{CPLEX}}\\
    \hline
         \textbf{Optimization problem} & Classical & Quantum & Classical & Quantum\\
         \hline
         SK model & \SKgurobi & $\tilde{O}(2^{0.247n})$ & \SKcplex & $\tilde{O}(2^{0.257n})$ \\
         \hline
         Maximum independent set & \MISgurobi & $\tilde{O}(2^{0.0187n})$ & \MIScplex & $\tilde{O}(2^{0.0170n})$\\
          \hline
         Portfolio optimization & \PORTgurobi & $\tilde{O}(2^{0.098n})$ & \PORTcplex & $\tilde{O}(2^{0.070n})$\\
         \hline
    \end{tabular}
    \caption{Complexity for different optimization problems using classical MIP solvers and the Quantum Branch and Bound algorithm corresponding to these solvers.}
    \label{table:numerics}
\end{table}

\begin{table}[h]
    \centering
    \begin{tabular}{|c||c|c|}
    \hline
         \textbf{Optimization problem} & \textbf{Gurobi} & \textbf{CPLEX} \\
         \hline
         SK model & $173\%$ & $187\%$ \\
         \hline
         Maximum independent set & $49\%$ & $56\%$\\
          \hline
         Portfolio optimization & $45263\%$ & $12587\%$\\
         \hline
    \end{tabular}
    \caption{Spread on the number of explored nodes for the largest problem sizes measured as $\max -\min$ expressed as a percentage of the median value.}
    \label{table:spread}
\end{table}

\begin{figure}[!ht]
     \centering
     \begin{subfigure}[b]{0.32\linewidth}
         \centering
         \includegraphics[width=\linewidth]{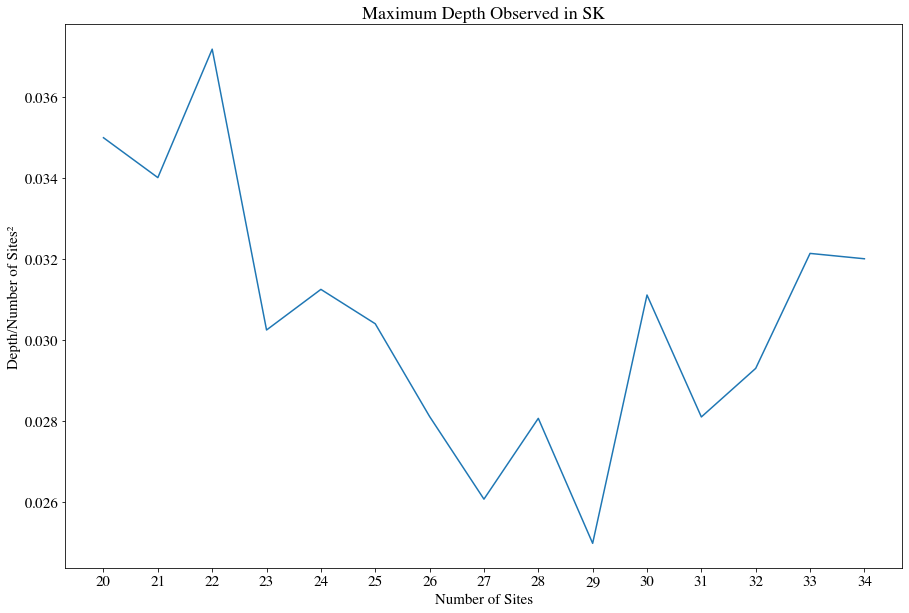}
         \caption{Depth to number of sites ratio for the SK problem}
         \label{fig:depth_dependency_SK}
     \end{subfigure}
     \hfill
     \begin{subfigure}[b]{0.32\linewidth}
         \centering
         \includegraphics[width=\linewidth]{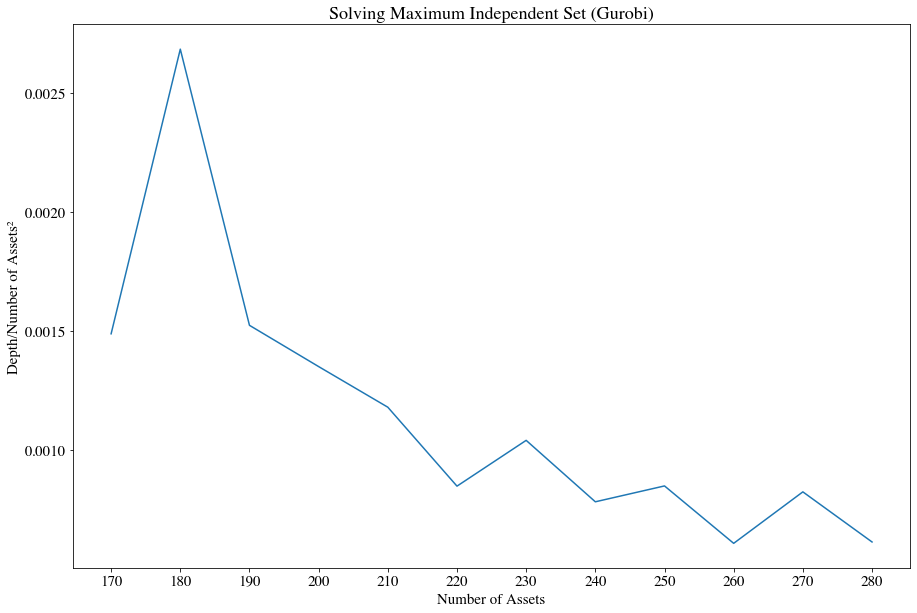}
         \caption{Depth to number of graph nodes ratio for MIS problems}
         \label{fig:depth_dependency_MIS}
     \end{subfigure}
     \hfill
     \begin{subfigure}[b]{0.32\linewidth}
         \centering
         \includegraphics[width=\linewidth]{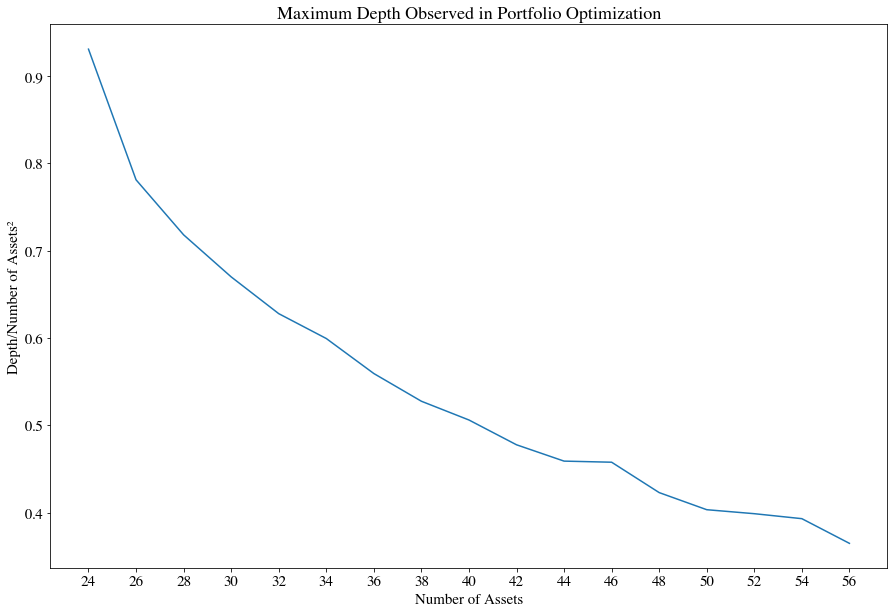}
         \caption{Depth to number of assets for portfolio optimization}
         \label{fig:depth_dependency_PORT}
     \end{subfigure}
     \hfill
     \caption{Depth $d$ of \bnb{} trees explored by Gurobi as a function of problem size $n$. We plot the ratio $d/n^{2}$ for three different problems, to establish that $d = O(n^2)$.} 
     \label{fig:depth_dependency}
\end{figure}

\subsection{Sherrington-Kirkpatrick Model}
\label{sec:numerics-sk}
For the SK model, we generated 50 random instances per problem size between 20 and 36 number of sites. Fig \ref{fig:SK} shows the number of nodes explored by classical optimizers as a function of the number of sites in the problem instance. With a linear regression over the median values we obtained a complexity of \SKgurobi~ with  $r^2>0.998$ for the results obtained with Gurobi and \SKcplex\ with $r^2>0.996$ for ones obtained with CPLEX.

\begin{figure}[!ht]
     \centering
     \begin{subfigure}[b]{0.45\linewidth}
         \centering
         \includegraphics[width=\linewidth]{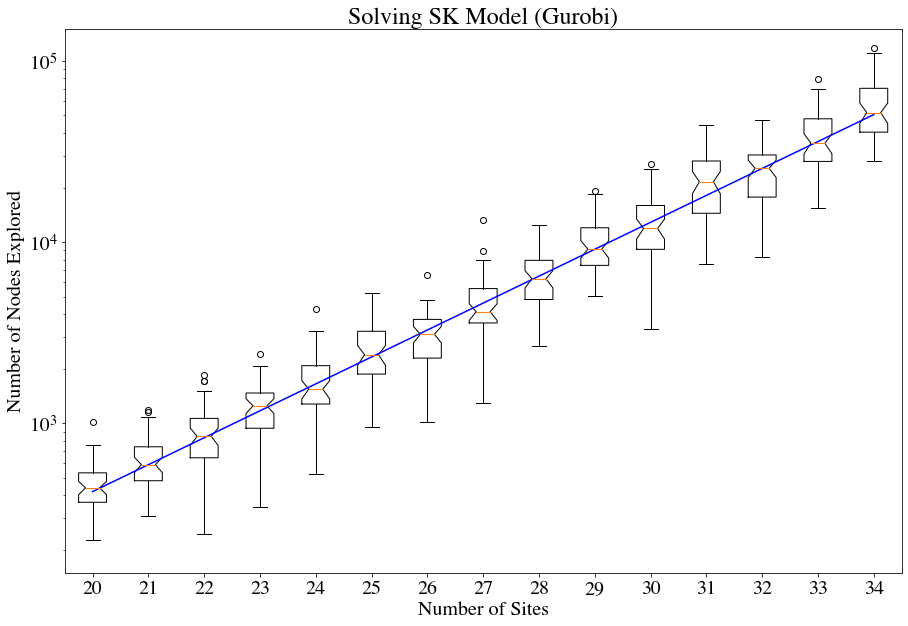}
         \caption{Estimated median complexity \SKgurobi~ with $r^2>0.998$ }
         \label{fig:SKgurobi}
     \end{subfigure}
     \hfill
     \begin{subfigure}[b]{0.45\linewidth}
         \centering
         \includegraphics[width=\linewidth]{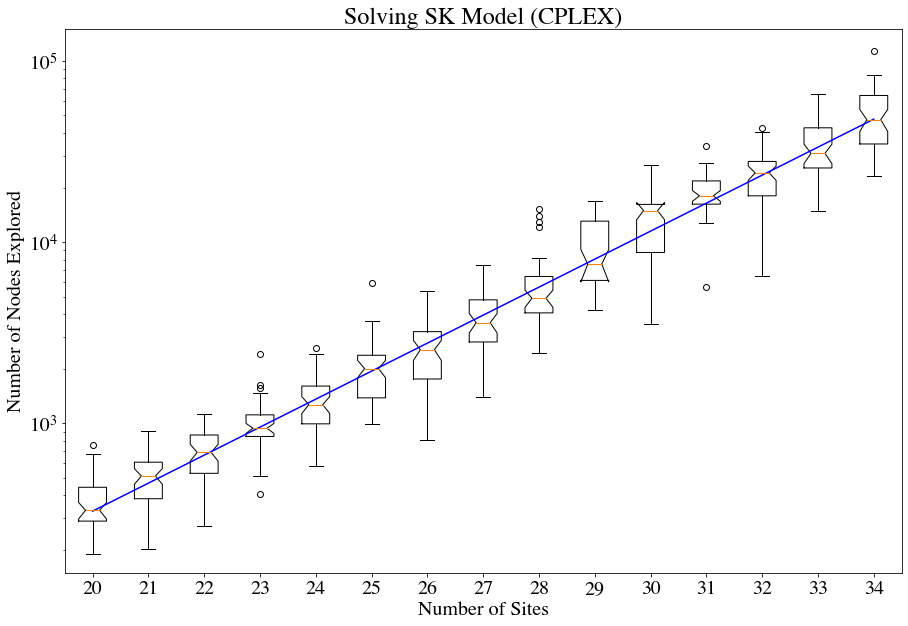}
         \caption{Estimated median complexity \SKcplex~ with $r^2>0.996$}
         \label{fig:SKcplex}
     \end{subfigure}
     \caption{Exponential evolution of the number of nodes explored by Gurobi (a) and CPLEX (b) as a function of the number of sites in the SK model. 50 random instances were considered for each problem characterized by the number of sites. The blue line indicates the linear regression made over the median value of nodes explored.}
     \label{fig:SK}
\end{figure}

\subsection{Maximum Independent Set}
\label{sec:numerics-mis}
For the MIS problem, we generated problems with a number of nodes ranging from $170$ to $290$. For each problem we randomly generated $20$ Erdős-Rényi graphs \cite{erdHos1960evolution} with a $0.8$ edge probability between two nodes. We optimized these problem instances with both Gurobi and CPLEX. In particular, the results obtained with CPLEX are plotted in Figure \ref{fig:MIS}, where the number of nodes explored by the algorithm grows exponential on the number of nodes in the graph. We obtained a measured complexity of \MIScplex~ with $r^2>0.989$. When using Gurobi, we obtained a very similar plot and a measured complexity of \MISgurobi~ with $r^2>0.996$. 

\begin{figure}[!ht]
     \centering
     \begin{subfigure}[b]{0.45\linewidth}
         \centering
         \includegraphics[width=\linewidth]{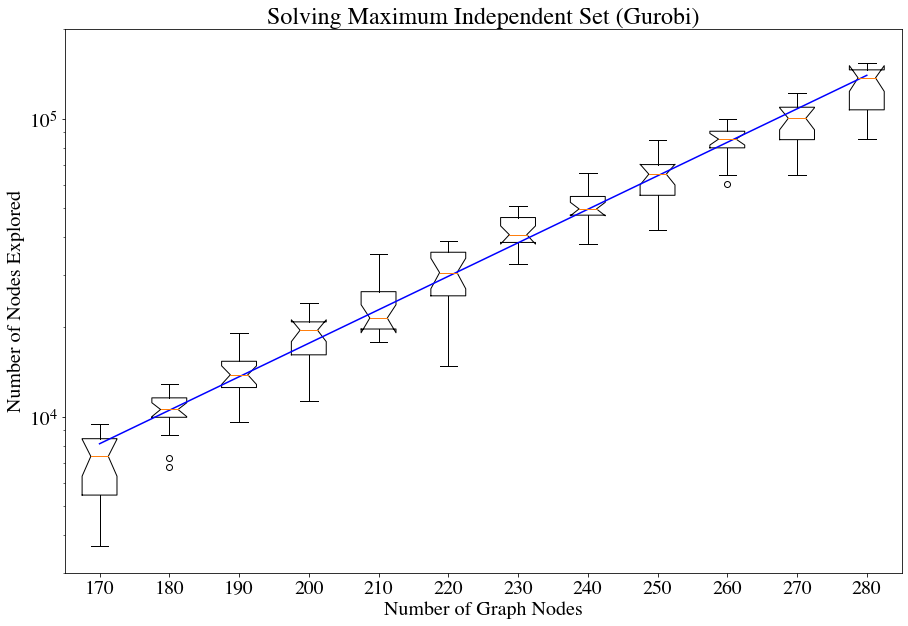}
         \caption{Estimated median complexity \MISgurobi~ with $r^2>0.996$ }
         \label{fig:MISgurobi}
     \end{subfigure}
     \hfill
     \begin{subfigure}[b]{0.45\linewidth}
         \centering
         \includegraphics[width=\linewidth]{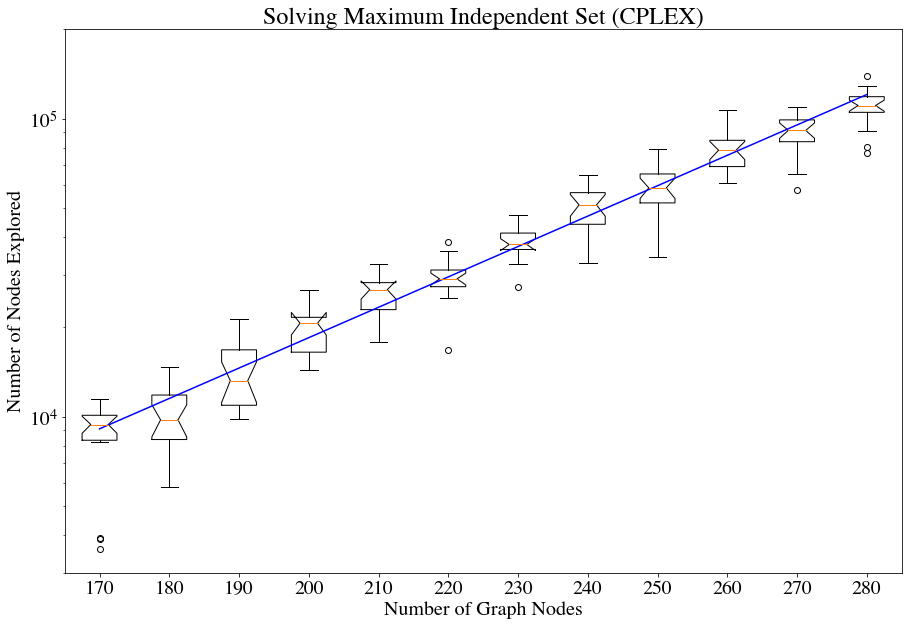}
         \caption{Estimated median complexity \MIScplex~ with $r^2>0.989$ }
         \label{fig:MIScplex}
     \end{subfigure}
      \caption{Exponential evolution of the number of nodes explored by Gurobi (a) and CPLEX (b) as a function of the graph size of the MIS instances. 20 random instances were considered for each problem characterized by the graph size. The blue line indicates the linear regression made over the median value of nodes explored.}
     \label{fig:MIS}
\end{figure}

\begin{figure}[!ht]
     \centering
     \begin{subfigure}[b]{0.45\linewidth}
         \centering
         \includegraphics[width=\linewidth]{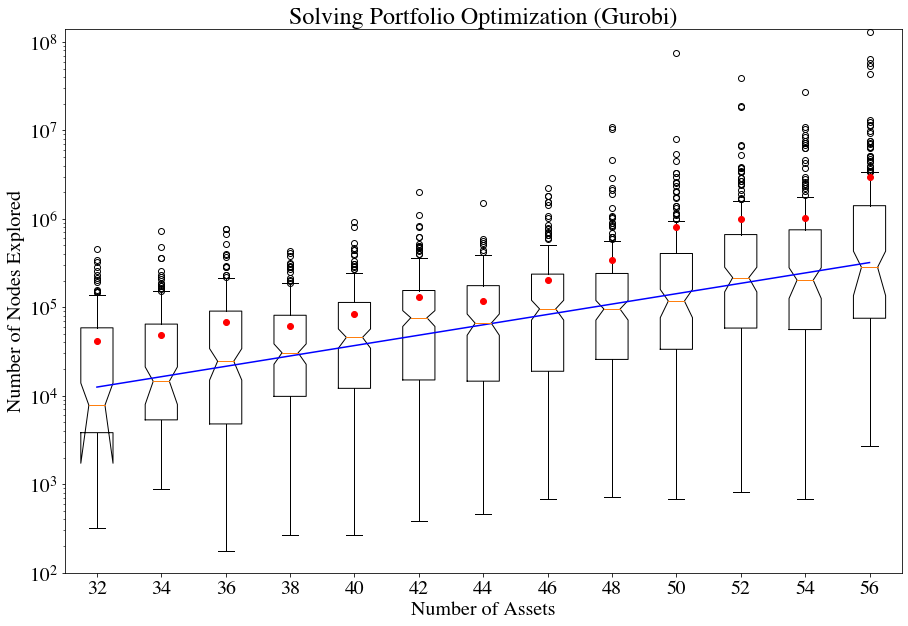}
         \caption{Estimated median complexity \PORTgurobi~ with $r^2>0.953$ }
         \label{fig:PORTgurobi}
     \end{subfigure}
     \hfill
     \begin{subfigure}[b]{0.45\linewidth}
         \centering
         \includegraphics[width=\linewidth]{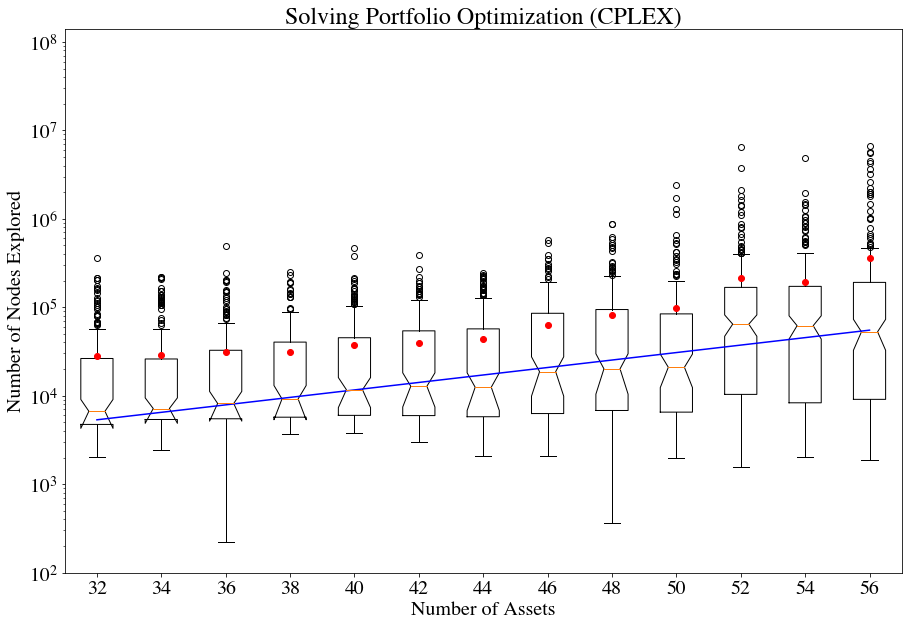}
         \caption{Estimated median complexity \PORTcplex~ with $r^2>0.898$ }
         \label{fig:PORTcplex}
     \end{subfigure}
     \caption{Exponential evolution of the number of nodes explored by Gurobi (a) and CPLEX (b) as a function of the number of assets in the universe. 200 random instances were considered for each problem characterized by the number of assets.  The red dots correspond to the mean number of nodes explored. The blue line indicates the linear regression made over the median value of nodes explored.}
     \label{fig:PORT}
\end{figure}

\subsection{Portfolio Optimization}
\label{sec:numerics-portfolio}
Last but not least, we consider a \textit{mean-variance portfolio-optimization problem} with constraints on the budget, the cardinality and the diversification. The problem is defined as follows, where $\mu$ is the historical returns, $\Sigma$ is the covariance matrix and $\Pi$ are the current prices.

\begin{equation*}
    \begin{split}\begin{aligned}
    \min_{x \in \N^{n-1} \times \R}  &q x^T \Sigma x - \mu^T x\\
    \text{subject to: } &\Pi^T x = B\\
    &||x||_0 = n/2\\
    \forall i, &\Pi_i x_i < 0.1B 
    \end{aligned}\end{split}
\end{equation*}

We generated 200 problem instances with a number of assets between 32 and 56. We randomly generated $\mu$, $\Sigma$ and $\Pi$. We optimized these instances with both CPLEX and Gurobi. The results obtained are plotted in Fig.\ref{fig:PORT}. As the means values can be up to an order of magnitude higher than the median values, we plotted the mean values using red dots. The complexity for Gurobi evolves as \PORTgurobi~ with $r^2>0.953$. We obtained similar results with CPLEX with a measured complexity of \PORTcplex~ with $r^2>0.898$.

We note that the problem of quantum algorithms for portfolio optimization has been studied in the continuous setting, i.e., with no integer variables~\cite{rebentrost2018portfolio,yalovetzky2021nisq,kerenidis19portfolio}. In this setting portfolio optimization can be solved in polynomial time by classical or quantum convex optimization algorithms, and the quantum algorithms obtain polynomial end-to-end speedups. Our algorithms extend these studies to the more general problem of portfolio optimization with integer variables which are often unavoidable, for e.g. in problems with diversification, or minimum transaction level constraints~\cite{cornuejols2006optimization}.

\bibliographystyle{unsrt}
\bibliography{main}

\section*{Acknowledgements}
The authors wish to thank Dylan Herman, Arthur Rattew, Ruslan Shaydulin, and Yue Sun for many helpful discussions and suggestions. Thanks also to the other members of the Global Technology Applied Research Center at JPMorgan Chase for their support and insights, and specially to Shaohan Hu and Chun-Fu (Richard) Chen for their \LaTeX\ wizardry.

\section*{Disclaimer}
This paper was prepared for information purposes with contributions from the Global Technology Applied Research Center of JPMorgan Chase. This paper is not a product of the Research Department of JPMorgan Chase or its affiliates. Neither JPMorgan Chase nor any of its affiliates make any explicit or implied representation or warranty and none of them accept any liability in connection with this paper, including, but not limited to, the completeness, accuracy, reliability of information contained herein and the potential legal, compliance, tax or accounting effects thereof. This document is not intended as investment research or investment advice, or a recommendation, offer or solicitation for the purchase or sale of any security, financial instrument, financial product or service, or to be used in any way for evaluating the merits of participating in any transaction.

\end{document}